\documentclass{sig-alternate}

\usepackage[T1]{fontenc}
\usepackage[latin9]{inputenc}
\usepackage{color}
\usepackage{verbatim}
\usepackage{float}
\usepackage{amsmath}

\usepackage{amsthm}
\usepackage{amssymb}
\usepackage{epsfig}
\usepackage{graphicx}
\usepackage{url}
\usepackage{subfigure}
\usepackage{bm}
\usepackage[normalem]{ulem}
\usepackage[ruled,commentsnumbered]{algorithm2e}
\usepackage{algpseudocode}

\theoremstyle{plain}
\newtheorem{thm}{\protect\theoremname}
\theoremstyle{definition}

\newtheorem{definition}{Definition}
\newtheorem{lemma}[thm]{\protect\lemmaname}
\newtheorem{prop}[thm]{\protect\propositionname}
\newtheorem{cor}{Corollary}
\newtheorem{theorem}{Theorem}

\providecommand{\theoremname}{Theorem}
\providecommand{\lemmaname}{Lemma}

\providecommand{\definitionname}{Definition}
\providecommand{\propositionname}{Proposition}
\ifodd 1272
\newcommand{\rev}[1]{{\color{blue}#1}} 
\newcommand{\com}[1]{\textbf{\color{red} (COMMENT: #1) }} 
\newcommand{\comg}[1]{\textbf{\color{green} (COMMENT: #1)}}
\newcommand{\response}[1]{\textbf{\color{green} (RESPONSE: #1)}} 
\else

\newcommand{\rev}[1]{{\color{black}#1}} 

\newcommand{\com}[1]{}
\newcommand{\comg}[1]{}
\newcommand{\response}[1]{}
\fi

\newfont{\mycrnotice}{ptmr8t at 7pt}
\newfont{\myconfname}{ptmri8t at 7pt}
\let\crnotice\mycrnotice%
\let\confname\myconfname%

\permission{Permission to make digital or hard copies of all or part of this work for personal or classroom use is granted without fee provided that copies are not made or distributed for profit or commercial advantage and that copies bear this notice and the full citation on the first page. Copyrights for components of this work owned by others than ACM must be honored. Abstracting with credit is permitted. To copy otherwise, or republish, to post on servers or to redistribute to lists, requires prior specific permission and/or a fee. Request permissions from permissions@acm.org.}
\conferenceinfo{SIGMETRICS'15,}{June 15--19, 2015, Portland, OR, USA.}
\copyrightetc{Copyright \copyright~2015 ACM \the\acmcopyr}
\crdata{978-1-4503-3486-0/15/06\ ...\$15.00.\\}

\begin{document}

\global\long\def\cond#1{\left|#1\right.}

\global\long\def\case#1#2{#1=\left\{  #2\right.}

\global\long\def\flr#1{\left\lfloor #1\right\rfloor }

\global\long\def\ceil#1{\left\lceil #1\right\rceil }

\title{Exchange of Services in Networks: Competition, Cooperation, and Fairness}

\numberofauthors{3}
\author{
\alignauthor
Leonidas Georgiadis\\ 
      \affaddr{Dept. of ECE, AUTH, Greece, leonid@auth.gr}\\
\alignauthor
George Iosifidis\\ \affaddr{Dept. of Elec. Eng., and YINS, Yale University, georgios.iosifidis@yale.edu}\\
\and
\alignauthor
Leandros Tassiulas\\   \affaddr{Dept. of Elec. Eng., and YINS, Yale University, leandros.tassiulas@yale.edu}\\
}
\vspace{-6mm}
\maketitle
\begin{abstract}
Exchange of services and resources in, or over, networks is attracting nowadays renewed interest. However, despite the broad applicability and the extensive study of such models, e.g., in the context of P2P networks, many fundamental questions regarding their properties and efficiency remain unanswered. We consider such a service exchange model and analyze the users' interactions under three different approaches. First, we study a centrally designed service allocation policy that yields the fair total service each user should receive based on the service it offers to the others. Accordingly, we consider a competitive market where each user determines selfishly its allocation policy so as to maximize the service it receives in return, and a coalitional game model where users are allowed to coordinate their policies. We prove that there is a unique equilibrium exchange allocation for both game theoretic formulations, which also coincides with the central fair service allocation. Furthermore, we characterize its properties in terms of the coalitions that emerge and the equilibrium allocations, and analyze its dependency on the underlying network graph. That servicing policy is the natural reference point to the various mechanisms that are currently proposed to incentivize user participation and improve the efficiency of such networked service (or, resource) exchange markets.
\end{abstract}



\section{Introduction}


\textbf{Motivation}. Today we are witnessing a renewed interest about models for exchanging services and resources in (or, over) networks, that go beyond the well-known peer-to-peer (P2P) file sharing idea. Some examples in communication networks are the WiFi sharing communities \cite{FON}, \cite{Sofia-UPN}, the mobile data sharing applications \cite{opengarden}, the commercial or community mesh networks \cite{confine}, \cite{bewifi}, and various peer-assisted services \cite{ioannidis-peer-assisted}, Fig. \ref{fig:system-model}. \rev{Similar schemes have been studied and implemented for other technological systems as well, e.g., for renewable energy sharing over smart grid \cite{gridmates}, \cite{saad-smart} where users share their energy surpluses with each other.} Finally, there is nowadays a plethora of online platforms, motivated by the sharing economy concept \cite{collaborativeconsumption}, which facilitate the exchange of commodities and services among users who, for example, are co-located, have common interests, etc, \cite{adalbdal}, \cite{getrridapp}, \cite{homeexchange}, \cite{neighborgood}, \cite{swapit}.

In essence, all the above scenarios apply the idea of collaborative consumption \cite{felson-coco-book} of underutilized resources (such as the Internet access) to networks with autonomous and self-interested nodes (or, users). \rev{Whenever a user has some idle resource, he offers it to other users who at that time have excess needs, and benefits in exchange from the resources they offer to him in the future \cite{georgiadis-netgcoop}. The goal is to exploit the nodes' complementarity in resource availability and demand, and increase the benefits for all the participants. Such models capture also more static settings where users have different preferences for the various resources, and exchange them in order to acquire those that are more valuable to them \cite{unver-book}.} There is a broad consensus that these models are of major importance for the economy, society and technological evolution \cite{collaborativeconsumption}, \cite{nytimes-sharing}. However, despite their significance and wide applicability, and although they have been subject to extensive research (e.g., in the context of P2P networks \cite{RJohariToNBilateral2011}, \cite{zhang-proportional}), some very important related questions remain unanswered.

\begin{figure}[t]
\vspace{-2mm}
\centering
\includegraphics[scale=0.7]{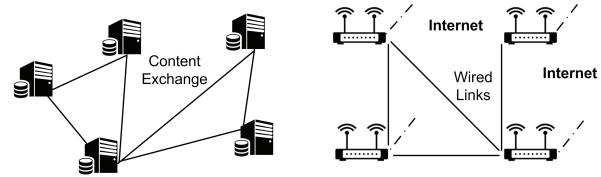}
\caption{\small{Instances of Service Exchange in Networks: a content sharing network, and an Internet sharing residential/community mesh network.}\label{fig:system-model}}
\vspace{-4mm}
\end{figure}

\emph{Definition and Properties of a Fair Exchange Policy.} This is one of the most critical issues in these cooperation schemes. Ideally, from a system design point of view, each user should receive service (or, resource) commensurate to its contribution. However, this is not always possible because there is an underlying graph that prescribes, for each user, the subset of the users it can serve and receive services from\footnote{\small{For example, in mesh networks the graph captures which nodes are within communication range, while in smart grid networks the graph shows which microgrids can exchange energy without significant transfer losses.}}. Additionally, there may be multiple feasible service exchange solutions that differ on the amount of service each user receives. We would prefer to select among them a \emph{fair} outcome that balances the exchanges as much as possible. The existence and the characterization of the properties of such fair policies (e.g., their dependency on the underlying graph) is an important and currently open question.

\emph{Existence and Fairness of Competitive Equilibriums}. Additionally, most often these systems are not controlled by a central entity that can exogenously impose such a fair solution. Instead, each user tries to greedily maximize its own benefit by allocating its idle resource to those users from which it expects to receive more service in return. A first question here is if such a competitive interaction among the nodes admits an equilibrium allocation, where each node cannot unilaterally improve the resource it accumulates. Also, we need to analyze how these equilibriums are affected by the graph structure and the nodes' resources. Finally, it is important to understand if such equilibriums are related to the centrally designed fair policy discussed above.

\emph{Robustness of the Fair Exchange Policy}. The latter question is related to the robustness of the fair policy: when a central designer proposes such a fair policy, is it possible for a user to deviate from it and improve his performance? More interestingly, in many cases it is possible to have a subset of users that deviate from the fair policy by forming a coalition and excluding non-members from bartering. For example, in a WiFi community a subset of users may decide to serve only each other, expecting that this will increase their benefits. Such strategies are very likely to deteriorate the overall system performance. A key challenge is to explore whether the fair policy is robust to such group deviations.


\textbf{Methodology and Contributions}. In order to shed light on these questions, we employ a general model that abstracts all the above scenarios. We consider a set of nodes, where each one has a certain idle resource that it allocates to its neighbors, and unsaturated demand for the resources of others. \rev{The model captures situations where the nodes have complementary resource availability over time, or generic static bartering markets where nodes simply have different preferences for the resources}. Neighborhood relationships are described by a bidirectional connected graph. The \emph{exchange ratio} (or, simply \emph{ratio}) of total received over allocated resource\footnote{\small{The idle resource can be the Internet bandwidth a user shares within a WiFi community during a month, the uploading capacity of a node in a P2P overlay. Similarly, the demand is the average request for additional Internet bandwidth (WiFi), the downloading capacity of a peer node, etc. Hereafter, we will use the term resource and service interchangeably.}} characterizes the performance of each node, as it quantifies the resource that it receives for each resource unit it offers.

From a system point of view, a central designer would prefer to have a vector of exchange ratios where each coordinate, that corresponds to a node, has value equal to one. Often this will not be possible due to the graph exchange constraints and asymmetries in nodes' resource availability. For that cases, the lexicographically maximum (lex-optimal), or max-min, exchange vector is a meaningful performance criterion as it is Pareto optimal and balances the exchanged resources as much as possible \cite{nace-tutorial}.

In the absence of a network controller however, we assume that each node makes greedy myopic allocation decisions so as to maximize the aggregate resource it receives in return. The interactions of the nodes give rise to a competitive market, which however differ from previous similar models \cite{osborne}, \cite{arrow-debreu}, \cite{zhang-proportional} due to the existence of the graph and the absence of side-payments (money) among the nodes (bartering). We introduce the concept of exchange equilibrium that is appropriate for this setting, characterize the equilibrium allocations, and study its relation to the max-min fair policy.

Accordingly, we assume that subset of  nodes can coordinate and form coalitions exchanging resources only with each other. A coalitional graph-constrained game with non-transferable utility (NTU) is identified in the above set-up. We focus on the existence and properties of stable equilibrium allocations. Given a certain global allocation, if there is a subset of nodes that when they reallocate their own resources among themselves manage to improve the exchange ratio of \emph{at least one} node in the subset, then they have an incentive to deviate from the global allocation (and hence destabilize it). Therefore, when an allocation is in equilibrium, it should be \emph{strongly stable} and no such subset should exist.

We study the above frameworks, that differ on the assumptions about the system control and the users behavior, and find a surprising connection among them. In particular:

(\textbf{i}) We prove that there is a unique equilibrium exchange ratio vector that is a solution for the competitive market, and lies in the core of the NTU graph-constrained coalitional game, being also strongly stable. This is the max-min fair (lex-optimal) ratio vector. It reveals that a centrally designed meaningful fair solution can be reached by nodes who act independently and selfishly, and it is also robust to group deviations. This finding has many implications for the applicability of such fair policies to decentralized and autonomous graph-constrained systems.

(\textbf{ii}) We show that the equilibrium exhibits rich structure and a number of interesting properties. For example, in the equilibrium allocation there is exchange of resources only among the nodes with the lowest exchange ratios and the nodes with the highest ratios, the nodes with the second lowest ratios with the set of the second highest ratios, and so on. We also study how the exchange ratios are affected by the graph properties, such as the node degree. This latter aspect is particularly important from a network design point of view as it reveals, among others, the impact a link removal or addition has on the equilibrium. Our findings hold for any graph, and therefore they can help a controller to predict or even dictate the exchange equilibrium.

(\textbf{iii}) We provide a polynomial-time algorithm that finds the lex-optimal exchange ratio vector and the resource exchange strategies that lead to it. Hence, it can also be used to find the equilibriums of the respective competitive and coalitional games. This is a highly non-trivial task in such exchange markets, that is further compounded here due to the graph constraints.

The rest of this paper is organized as follows. In Sec. \ref{sec:Model-Notation} we introduce the model, and in Sec. \ref{sec:Central-design} we prove the existence and analyze the properties of the lex-optimal exchange policy. In Sec. \ref{sec:game-theory-frameworks} we define and solve the coalitional and the competitive games. In Sec. \ref{sec:algorithms} we provide a polynomial algorithm for computing the lex-optimal ratios. We present several numerical examples in Sec. \ref{sec:Numerical-Results}. In Sec. \ref{sec:Related} we discuss related works, and conclude in Sec. \ref{sec:Conclusions}. In the Appendix we provide the additional proofs. 

\vspace{-1mm}
\section{Model and Problem Statement} \label{sec:Model-Notation}

We consider a service exchange market that is modeled as an undirected connected graph $G=({\cal N},{\cal E})$ with node and edge set ${\cal N}$ and ${\cal E}$, respectively. Let $\mathcal{N}_{i}=\{j\,:\,(i,j)\in\mathcal{E}\}$ be the set of neighbors of node $i\in{\cal N}$, and $D_i>0$ its idle resource (endowment). Let $d_{ij}\geq0$ be the resource that node $i$ allocates to node $j\in{\cal N}_{i}$. We assume that each non-isolated node allocates all its (idle) resource\footnote{\small{\rev{For example, in P2P overlays, each node allocates all its uplink bandwidth, and in other settings it exchanges resources for which it has zero valuation (e.g., excess food).}}}, i.e.,
\begin{equation}
\sum_{j\in{\cal N}_{i}}d_{ij}=D_{i},\,\,\forall\, i\in{\cal N},\,\,\mathcal{N}_i\neq \emptyset.\label{eq:Allloc}
\end{equation}
A vector $\bm{d}=(d_{ij}) _{(i,j)\in{\cal E}}$ satisfying (\ref{eq:Allloc}) is called ``allocation''. The set of allocations is denoted by $\mathbb{D}$. Note that as long as not all nodes are isolated, i.e., $\mathcal{E}\neq \emptyset$, it holds $\mathbb{D}\neq\emptyset$.

\rev{This model captures either (i) a static setting where each user has a certain amount of a perfectly divisible resource which wishes to trade with other, more valuable to him, resources, or (ii) a dynamic setting where users have at random time instances a single unit of unsplittable excess resource which they allocate to one of their neighbors expecting similar benefits in the future. This latter dynamic setting will become more clear in the sequel.}

A vector $\bm{r}$ of received resources induced by an allocation $d\in\mathbb{D}$ is called feasible. The set of feasible received resource vectors when $\mathcal{E} \neq \emptyset$ is defined as:
\begin{equation}
\mathbb{R}=\big\{ \bm{r}=(r_{i}) _{i\in{\cal N}}:\ r_{i}=\sum_{j\in{\cal N}_{i}}d_{ji},\ i\in{\cal N},\ \bm{d}\in\mathbb{D}\big\}, \label{eq:RateSpace}
\end{equation}
where we adopt the convention that for any isolated node $i$ it is $r_i=0$. In case $\mathcal{E} = \emptyset$, we define $\mathbb{R}=\{\boldsymbol{r}:\ r_i=0, i\in \cal{N}\}$.

Throughout this work, we will be interested in the \emph{exchange ratio} (or, simply \emph{ratio}) vector $\bm{\rho}=(\rho_{i}=r_i/D_i: i\in\mathcal{N})$, where the $i^{th}$ coordinate quantifies the aggregate amount of resource that node $i$ receives per unit of resource that offers to its neighbors. \rev{Notice that, under assumption (\ref{eq:Allloc}), maximizing $\rho_i$ ensures the maximization of $r_i$}. We denote by $\mathbb{P}$ the set of all feasible ratio vectors. In the sequel, we consider three different problem formulations based on the above model.

\subsection{Fairness Framework}

\rev{In this setting, the total allocated resource is always equal to $\sum_{i\in\mathcal{N}}D_i$, and therefore the various allocations differ on how they split this amount across the nodes.} A centrally designed policy for this cooperative setting would ideally allocate to every node $i\in\mathcal{N}$ resource equal to its contribution, i.e., $r_i=D_i$. However, due to the graph that constraints resource exchanges, and the different resource endowments of the nodes, such policies will not be realizable in general. Given this, the designer would prefer to ensure the most balanced allocation.

A suitable method to achieve this goal is to employ the lexicographic optimal (or, lex-optimal) criterion, which has been extensively used for resource allocation and load balancing in communication networks \cite{georgiadislexopt02}, \cite{nace-tutorial}, \cite{boudecfairness07}. This multi-objective optimization method first increases as much as possible the allocated resource to the node with the smaller exchange ratio. Next, if there are many choices, it attempts to increase the resource allocated to the node with the second smaller exchange ratio, and so on. The resulting allocation is max-min fair, thus as balanced as possible. Next we provide the necessary definitions.
\vspace{-1mm}
\begin{definition}
\textbf{Lexicographical order}. Let $\bm{x}$ and $\bm{y}$ be $N$-dimensional vectors, and $\phi(\bm{x})$ and $\phi(\bm{y})$ the respective $N$-dimensional vectors that are created by sorting the components of $\bm{x}$ and $\bm{y}$ respectively, in non-decreasing order. We say that $\bm{x}$ is lexicographically larger than $\bm{y}$, denoted by $\bm{x}\succ\bm{y}$, if the first non-zero component of the vector $\phi(\bm{x})-\phi(\bm{y})$ is positive. The notation $\bm{x}\succeq\bm{y}$ means that either $\bm{x}\succ\bm{y}$ or, $\bm{x}=\bm{y}$.
\end{definition}
\vspace{-1mm}
It is easy to see that the set of received resource vectors $\mathbb{R}$ defined in (\ref{eq:RateSpace}) is compact and convex. Hence, as is shown in \cite{boudecfairness07}, there is a unique lex-optimal $\bm{r}^{*}\in\mathbb{R}$ such that, with $\bm{\rho}^{*}=\left(r_{i}^{*}/D_{i}\right)_{i\in{\cal N}}$ and for any $\bm{\rho}=\left(r_{i}/D_{i}\right)_{i\in{\cal N}}\in \mathbb{P}$, it holds $\bm{\rho}^{*}\succeq\bm{\rho},\ \forall\,\bm{r}\in\mathbb{R}$. We are also interested in the respective \emph{lex-optimal} allocations $\bm{d}$, which are those that result in the unique lex-optimal ratio vector. Note that there may be \emph{many} allocations $\bar{\bm{d}}$ for which $\bar{\bm{\rho}}=\bm{\rho}^{*}$, as shown in Fig. \ref{fig:Def2-example}. \emph{Within this framework, we are interested in studying the properties of the lex-optimal exchange ratio vector, and the respective lex-optimal allocations, for any graph $G=(\mathcal{N},\mathcal{E})$ and any endowments $\{D_i\}_{i\in\mathcal{N}}$}.

\begin{figure}[t]
\centering
\includegraphics[scale=0.6]{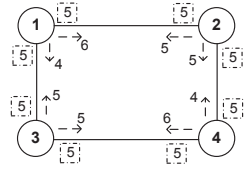}
\caption{\small{A 4-node graph with node resource $D_{i}=10$, $\forall i$. Dashed arrows indicate resource exchange under two different allocations (values shown internal and external of the graph). Both allocations yield ratios $\rho_i=1$, $\forall i$.} \label{fig:Def2-example}}
\end{figure}

\subsection{Coalitional Framework}

Before providing the details of this framework, let us introduce some additional notation. We denote by $G_{{\cal S}}=\left({\cal S},{\cal E_{{\cal S}}}\right)$ the subgraph of $G$ induced by a nonempty set of nodes $\mathcal{S}\subseteq\mathcal{N}$, i.e., the graph with node set ${\cal S}$, and edge set the edges $(i,j)\in{\cal E}$, with $i,j\in\mathcal{S}$. By ``allocation on ${\cal S}$'', we mean a vector $\bm{d}_{{\cal S}}=\left\{ d_{ij}\right\} _{(i,j)\in{\cal E}_{{\cal S}}}$ defined on graph ${G}_{\mathcal{S}}=\left(\mathcal{S},\mathcal{E}_{\mathcal{S}}\right)$ satisfying (\ref{eq:Allloc}) (with $\mathcal{N}\leftarrow\mathcal{S}$ and ${\cal E}\leftarrow{\cal E}_{{\cal S}}$). We denote by $\mathbb{D}_{{\cal S}}$ the set of all allocations on ${\cal S}$, and by $\mathbb{R}_{{\cal S}}$ the set of all received resource vectors on ${\cal S}$ which can be obtained by any allocation on $\mathcal{S}$. By definition it is $\mathbb{R}_{\{i\}}=\left\{ 0\right\}$, $\forall\,i\in{\cal N}$.

We assume here that the nodes are able to coordinate with each other, they can form coalitions and deviate from the proposed fair solution if this will ensure higher resources for one or more of them. In game theoretic terms, this behavior leads to a coalitional (or, cooperative) game \cite{myerson-gametheory-book} played by the nodes. Specifically, we call any subset of nodes ${\cal S\subseteq\mathcal{N}}$ a coalition when they allocate their resources only among each other. That is, there is no resource exchange among nodes in $\mathcal{S}$ and nodes in its complement set $\mathcal{S}^{c}=\mathcal{N}-\mathcal{S}$. Hence, the feasible resource vectors that nodes in ${\cal S}$ get, are the vectors of the set $\mathbb{R}_{{\cal S}}$. We also refer to the set $\mathcal{N}$ as the grand coalition. This coalitional game is one with non-transferable utilities, as the received resource vector $\bm{r}$ cannot be split arbitrary among the nodes, due to the exchange constraints imposed by the graph and the lack of side payments. More formally, we define \cite{myerson-gametheory-book}:
\vspace{-1mm}
\begin{definition} \textbf{Coalitional Service Exchange Game}. A non-transferable utility (NTU) game in graph form consists of the triplet $<\mathcal{N},G,\{\mathbb{R}_{\cal{S}},\mathcal{S}\in\mathcal{N}\}>$, where $\mathcal{N}$ is the set of players (nodes) with initial resource endowments $\left\{ D_{i}\right\} _{i\in\mathcal{N}}$, and $G=(\mathcal{N},\mathcal{E})$ is the graph describing the service exchange possibilities among the nodes. Moreover, $\mathbb{R_{{\cal S}}}$, $\,\mathcal{S}\subseteq\mathcal{N}$, is the set of feasible $|\mathcal{S}|$-dimensional vectors of players' received resources $\{\bm{r}_{i}\}_{i\in\mathcal{S}}$, satisfying properties (i) $\mathbb{R}_{\{i\}}=\left\{ 0\right\} ,\,\forall i\in\mathcal{N}$, (ii) $\mathbb{R_{{\cal S}}}$ is closed and bounded.
\end{definition}
\vspace{-1mm}
Our goal is to study the existence and the properties of self-enforcing allocations. This property is formally captured by the notion of stability for the grand coalition.
\vspace{-1mm}
\begin{definition} \textbf{Stability}. An allocation $\bm{d}$, and the respective resource vector $\bm{r}$ is called \emph{strongly} stable if for any node set ${\cal S}\subseteq{\cal N}$, there is no allocation $\widehat{\bm{d}}_{{\cal S}}$ on the induced subgraph $G_{{\cal S}}=\left({\cal S},{\cal E_{{\cal S}}}\right)$, such that $\hat{r}_{i}\geq r_{i}$ for all $i\in{\cal S}$, and $\hat{r}_{j}>r_{j}$ for at least one node $j\in{\cal S}$. The allocation is called \emph{weakly} stable if for any node set ${\cal S}\subseteq{\cal N}$ , there is no allocation $\widehat{\bm{d}}_{{\cal S}}$ such that $\hat{r}_{i}>r_{i}$ for all $i\in{\cal S}$.
\end{definition}
\vspace{-1mm}
Note that strong stability implies weak stability but not the other way around. In particular, the concept of weak stability for the grand coalition is directly related to the concept of the \emph{core} which is formally defined\footnote{\small{With a slight abuse of terminology we refer both to the received resource vectors and the respective allocations as stable.}} \cite{myerson-gametheory-book}:
\vspace{-1mm}
\begin{definition}\textbf{Core}. Given an NTU coalitional game $<\mathcal{N},G,\{\mathbb{R}_{\cal{S}},\ \cal{S}\in\cal{N}\}>$, the \emph{core} of $\mathbb{R}$ is defined as the subset of $\mathbb{R}$ which consists of all received resource vectors $\bm{r}\in\mathbb{R}_{{\cal }}$, such that for any possible coalition $\mathcal{S}$ and any allocation $\hat{\bm{d}}\in\mathbb{D}^{|\mathcal{S}|}$, if $\hat{r}_{i}>r_{i}$, for all $i\in\mathcal{S}$, then $\hat{\bm{r}}\notin\mathbb{R_{{\cal S}}}$.
\end{definition}
\vspace{-1mm}
In this coalitional framework, we ask the question: \emph{Is there a (weakly or strongly) stable allocation for this service exchange coalitional game, and if yes, what are its properties in terms of allocations and exchange ratios}?

\subsection{Competitive Framework}

Assume now that each node $i\in\mathcal{N}$ is an independent decision maker, devising its allocation vector $\bm{d}_i=\big(d_{ij}\big)_{j\in\mathcal{N}_i}$ so as to maximize the resource $r_i$ it receives. In such a competitive market setting, the nodes are allowed to select any policy that satisfies eq. (\ref{eq:Allloc}), i.e., allocating their entire resource (market clearing condition). \rev{Namely, the solution concept for this market is related to the competitive (or, Walrasian) equilibrium \cite{arrow-debreu}, \cite{ColellWhinstonGreenBook1995}, which has been also applied in communication networks \cite{RJohariToNBilateral2011}, and extended to graphical economies (which exhibit \emph{localities}) \cite{KearnsGraphEcon2004}, \cite{KearnsEconSocial2004}. However, for the problem under consideration, there do not exist explicit price variables (or, price signals), and hence we employ a different equilibrium concept}:
\vspace{-1.5mm}
\begin{definition}
\textbf{\emph{Exchange Equilibrium.}}  An allocation $\bm{d}$ is an exchange equilibrium, if and only if (iff) for any node $i\in \mathcal{N}$ it holds (i) $d_{ji}=\rho_{i}d_{ij}$ for all $j\in\mathcal{N}_{i}$, and (ii) if $d_{ji}>0$ for some $j\in\mathcal{N}_{i}$ then $\rho_{j}=\min_{k\in\mathcal{N}_{i}}\rho_{k}$.
\end{definition}
\vspace{-1mm}
In other words, at the equilibrium each node $i$ exchanges services only with its neighbors that trade in the lowest exchange ratio, so as to receive the maximum possible total service. Additionally, all the nodes that interact with $i$, have the same exchange ratio, while there may exist neighbors that do not allocate any resource to it. These latter nodes will certainly have higher exchange ratios, i.e., reciprocate with less resource for each unit of resource they receive. In this context, the question we want to tackle is the following: \emph{Does this game have exchange equilibrium(s), and if so, what are its properties and how does it depend on the graph $G$}.



\section{Lex-optimal Allocations}\label{sec:Central-design}


%

In this section we study the properties of the lex-optimal vector $\bm{\rho}^{*}$ and the respective allocations. These results hold for any graph $G=(\mathcal{N},\mathcal{E})$, and resource endowments $\{D_i\}_{i\in\mathcal{N}}$. To avoid trivial cases, we assume that there are no isolated nodes in the network\footnote{\small{If a graph $G$ has a set of isolated nodes $\mathcal{I}$ then we set $\rho_i =0$ for all $i\in \mathcal{I}$ and we proceed by considering the graph $G_{\mathcal{N}-\mathcal{I}}$.}}. We give first some notations. We denote the set of neighbors of nodes in a set $\mathcal{S}$, that do not belong to $\mathcal{S}$, by $\mathcal{N}\left(\mathcal{S}\right)=\cup_{i\in\mathcal{S}}\mathcal{N}_{i}-\mathcal{S}$. Given an allocation $\bm{d}$, for each node $i\in{\cal N}$ we define the subset $\mathcal{D}_{i}=\left\{ j\in\mathcal{N}_{i}:\, d_{ij}>0\right\}$ of nodes that receive resource from $i$, and the subset of nodes that don't receive resource ${\cal H}_{i}={\cal N}_{i}-{\cal D}_{i}=\left\{ j\in\mathcal{N}_{i}:\, d_{ij}=0\right\}$. Also, we define the subset $\mathcal{R}_{i}=\left\{ j\in\mathcal{N}_{i}:\, d_{ji}>0\right\}$ of nodes that give resource to $i$.


For a given $\bm{r}$, the set of \emph{different values (levels)} the coordinates of vector $\bm{\rho}$ take, will be denoted by $l_{k},\ i=1,...,K\leq N$, where $l_{1}<l_{2}<...<l_{K}$. The index of the level to which $\rho_{i}$ is equal, is denoted by $k(i)$, i.e., $l_{k(i)}=\rho_{i}$. We call $k(i)$ the ``level of node $i$''. The set of nodes with level $m$ is denoted by $\mathcal{L}_{m}=\left\{ i\in{\cal N}:\ k(i)=m\right\}$. If a subset of nodes $\mathcal{S\subseteq\mathcal{N}}$ has the same level under an allocation $\bm{d}$, then we denote the index of this level by $k\left({\cal S}\right)$. Note that the above quantities depend on the allocation $\bm{d}$, and hence, in order to facilitate notation, we will use the same overline symbol for them whenever applicable.

\textbf{Properties}. An important well-known property of the lex-optimal policy is that it is \emph{Pareto efficient} \cite{ColellWhinstonGreenBook1995}, \cite{nace-tutorial}, i.e., we cannot increase the exchange ratio for one node without decreasing the ratio of another node. The first property of the lex-optimal allocations that we prove is that the neighbors of each node $i\in\mathcal{N}$, that receive non-zero resource from $i$, belong to the same exchange ratio level set\footnote{\small{Recall that this ratio is determined by the total resource each of these nodes receives, i.e., not only from that allocated by node $i$.}}. Moreover, all the neighbors that do not receive resource from $i$, have a higher level index. Specifically:
\vspace{-1mm}
\begin{lemma}
\label{lem:neighbor}
Let $\bm{d}^{*}$ be a lex-optimal allocation, and let $i\in{\cal N}$. Then all nodes $j\in\mathcal{D}_{i}^{*}$ have the same level $l_{k(\mathcal{D}_{i}^{*})}^{*}$ and hence belong to the same set $\mathcal{L}_{k(\mathcal{D}_{i}^{*})}^{*}$. Moreover, for any node $j\in \mathcal{H}_{i}^{*}$, it is $l_{k(j)}^{*}\geq l_{k(\mathcal{D}_{i}^{*})}^{*}$.
\vspace{-1mm}
\end{lemma}
\begin{proof}
Consider a lex-optimal allocation $\boldsymbol{\bar{d}}$ and let $j_{1},\ j_{2}$ be such that $\bar{d}_{ij_{1}}>0$, $\bar{d}_{ij_{2}}>0$, but $j_{1}\in{\cal L}_{m}^{*}$ and $j_{2}\in\mathcal{L}_{n}^{*}$ with $m<n$. Recall that the lex-optimal price vectors - and hence the respective level sets - are unique so we use the star ($*$) notation for them. We can then move some resource from $j_{2}$ and give it to $j_{1}$ while ensuring that with the resulting allocation $\widehat{\boldsymbol{d}}$, it is $l_{k(j_{1})}^{*}<\widehat{l}_{k(j_{1})}\leq\widehat{l}_{k(j_{2})}<l_{k(j_{2})}^{*}$. Since the received resources of all other nodes remain the same, it follows that $\widehat{\boldsymbol{r}}\succ\boldsymbol{\bar{r}}=\boldsymbol{r}^{*}$, which is a contradiction as we assumed that it is lex-optimal. Assume next that $d_{ij}^{*}=0$ for a node $j\in\mathcal{N}_{i}$ for which it holds that $l_{k(j)}^{*}<l_{k(\mathcal{\bar{D}}_{i})}^{*}$. Using a similar argument we arrive at a contradiction again. $\blacksquare$
\end{proof}
\noindent Note that this lemma shows already that lex-optimal allocations have some of the required properties of exchange equilibriums. As will be  shown later there are lex-optimal allocations that are in fact exchange equilibriums.

\begin{figure}[t]
\vspace{-1mm}
\begin{centering}
\includegraphics[scale=0.48]{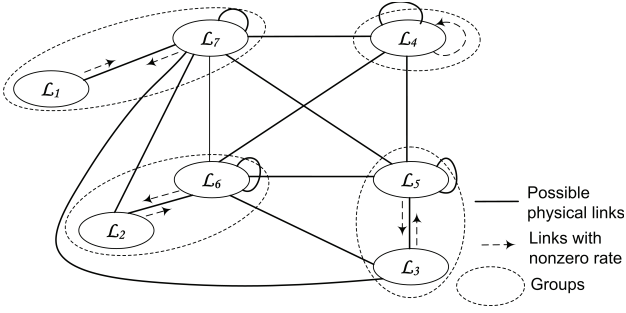}
\par
\end{centering}
\caption{\small{Structure of a graph with $K^{*}=7$ levels.}\label{fig:Corollary-Example}}
\vspace{-1mm}
\end{figure}

We need some additional notation at this point. Let $\bm{d}\in\mathbb{D}$ and assume $K\geq2$, i.e., the allocation has at least two levels. Define ${\cal Q}_{1}={\cal N}$, and for $K\geq3$:
\begin{align*}
{\cal Q}_{k} & ={\cal N}-\cup_{m=1}^{k-1}\left({\cal L}_{m}\cup{\cal L}_{K-m+1}\right),\ 2\leq k\leq\lceil{K/2\rceil}.
\end{align*}
For example, $\mathcal{Q}_{2}$ consists of the nodes in $\mathcal{N}$ that remain after removing those that belong to the level sets $\mathcal{L}_1$ and $\mathcal{L}_K$. In the sequel, a quantity $X$ referring to induced subgraph $G_{{\cal Q}_{k}}=\left({\cal Q}_{k},{\cal E}_{{\cal Q}_{k}}\right)$ is denoted $X_{{\cal Q}_{k}}$. The next Theorem describes the structure of a lex-optimal allocation.
\begin{theorem}
\label{thm:MainTh0}
If an allocation $\bm{d}^{*}$ is lex-optimal and $K\geq2$, then the following Properties hold:
\begin{enumerate}
\item \label{enu:MainTh0Item1}${\cal L}_{k}^{*}$ is an independent set in graph $G_{{\cal Q}_{k}}$, for $k=1,....,\lfloor{\frac{K}{2}\rfloor}$.
\item \label{enu:MainTh0Item2}${\cal L}_{K-k+1}^{*}={\cal N}_{{\cal Q}_{k}}\left({\cal L}_{k}^{*}\right)$, for $k=1,....,\lfloor{\frac{K}{2}\rfloor}$.
\item \label{enu:MainTh0Item2.1}$l_{k}^{*}l_{K-k+1}^{*}=1$, for $k=1,....,\lfloor{K/2\rfloor}$.
\item \label{enu:MainTh0Item3}$\sum_{i\in{\cal L}_{k}^{*}}r_{i}^{*}=\sum_{i\in{\cal L}_{K-k+1}^{*}}D_{i}$, for $k=1,....,\lfloor{\frac{K}{2}\rfloor}$.
\end{enumerate}

\end{theorem}
\vspace{-1.5mm}

Interestingly, sufficiency of this result holds as well:
\vspace{-1mm}
\begin{theorem}
\label{thm:MainLexSuff}
If an allocation $\bm{d}$ with $K\geq2$ has properties \ref{enu:MainTh0Item1}-\ref{enu:MainTh0Item3} of Theorem \ref{thm:MainTh0}, then it is lex-optimal.
\end{theorem}
\vspace{-1mm}
Finally, when there is only one level set, it holds:
\vspace{-1mm}
\begin{theorem}
\label{thm:MainTh0_K1}
If an allocation $\bm{d}^{*}$ is lex-optimal, and $K^{*}=1$, then $l_{1}^{*}=1$. Also, if an allocation $\bm{d}$ has $K=1$, then $l_{1}=1$ and it is lex-optimal.
\end{theorem}
\vspace{-1mm}

\subsection{Analysis and Discussion}

Let us now discuss the implications of the above theorems. Under a lex-optimal allocation, the nodes are classified in disjoint sets of different levels, in a fashion that depends both on their resource endowments and on the graph $G$. For the discussion below, please refer to Fig. \ref{fig:Corollary-Example}, that presents an example of the structure for $K^{*}=7$ levels. In this graph, we depict with solid lines the physical connections that may exist among the different sets of nodes. Notice that the actual nodes and their detailed connections are not shown.


\uline{Exchange ratios Structure}. According to Property $3$ of Theorem $1$, the exchange ratios have a certain structure. Specifically, the highest level of ratios is inversely proportional to the lowest level of ratios ($l_{7}^{*}=1/l_{1}^{*}$), the second highest exchange ratio level is inversely proportional to the second lowest ratio level ($l_{6}^{*}=1/l_{2}^{*}$), and so on. Additionally, resource exchanges satisfy Property $4$. For example, in Fig. \ref{fig:Corollary-Example} all the nodes of the highest ratio set $\mathcal{L}_{7}^{*}$, allocate their entire resources to the nodes belonging to the lowest level set $\mathcal{L}_{1}^{*}$. Moreover, the latter receive resource only from the nodes in $\mathcal{L}_{7}^{*}$. Similarly, the nodes in level set $\mathcal{L}_{6}^{*}$ allocate all their received resources to nodes in set $\mathcal{L}_{2}^{*}$ which are served only by these former nodes, and so on. Interestingly, when $K^{*}$ is odd, there is one set of nodes, here the set $\mathcal{L}_{4}^{*}$, which exchange resource only with each other.

\uline{Topological Properties}. On the other hand, Properties $1$ and $2$ reveal the \emph{impact of the network topology on the max-min solution}. First, from Property $2$, we can find the possible neighbors for each node, based on the level set it belongs to. For example, when $K^{*}=7$, it holds $\mathcal{L}_{7}^{*}=\mathcal{N}(\mathcal{L}_{1}^{*})$, i.e., the neighbors of nodes in $\mathcal{L}_{1}^{*}$, that do not have ratios $l_{1}^{*}$, belong only to set $\mathcal{L}_{7}^{*}$. Moreover, since Property $1$ states that the set $\mathcal{L}_{1}^{*}$ in independent in the graph $G_{{\cal Q}_{1}}\triangleq G$, we understand that $\mathcal{L}_1$ nodes have neighbors only in $\mathcal{L}_7$.

Similarly, it holds that $\mathcal{L}_{6}^{*}=\mathcal{N}_{\mathcal{Q}_{2}}(\mathcal{L}_{2}^{*})$. Hence, the nodes in set $\mathcal{L}_{2}^{*}$ can have links only with nodes in set $\mathcal{L}_{6}^{*}$ and possibly with nodes in $\mathcal{L}_{7}^{*}$ (since the latter do not belong in $G_{\mathcal{Q}_2}$). However, from Lemma \ref{lem:neighbor} it follows that nodes in $\mathcal{L}_{2}^{*}$ exchange resource only with nodes in $\mathcal{L}_{6}^{*}$. With the same reasoning, it is easy to see that nodes in set $\mathcal{L}_{3}^{*}$ can be physically connected with nodes in $\mathcal{L}_{7}^{*}$, $\mathcal{L}_{6}^{*}$ and $\mathcal{L}_{5}^{*}$, but they exchange resource only with nodes in the latter set. Finally, nodes in set $\mathcal{L}_{4}^{*}$ exchange resources only with each other.


These properties reveal how the graph affects the lex-optimal fair solution. For example, by adding a link between two nodes initially belonging to $\mathcal{L}_{1}^{*}$ (which is independent), the lex-optimal solution changes and places these (now connected) nodes to another set. This dependency among the graph structure and the lex-optimal exchange ratio vector will become more evident in the sequel.



\subsection{Proofs of Theorems}
In this subsection we provide the proofs of Theorems \ref{thm:MainTh0_K1} and \ref{thm:MainTh0}, while Theorem \ref{thm:MainLexSuff} is proved in the Appendix.

\uline{\textbf{3.2.1 PROOF of Theorem \ref{thm:MainTh0_K1}}}. Before proceeding with the proof, we provide some additional notation, lemmas and propositions. We denote the sum of received resources that are incoming to, and outgoing from set ${\cal S}\subseteq{\cal N}$, under allocation $\boldsymbol{d}$, as follows:
\begin{equation}
{\rm In}\left({\cal S}\right)=\sum_{i\in{\cal S}}\ \sum_{j\in{\cal N}_{i}\cap{\cal S}^{c}}d_{ji},\,\,\,{\rm Out}\left({\cal S}\right)=\sum_{i\in{\cal S}}\ \sum_{j\in{\cal N}_{i}\cap{\cal S}^{c}}d_{ij}, \label{eq:In-Out-3}
\end{equation}
where $\mathcal{S}^{c}=\mathcal{N}-\mathcal{S}$ is the complement set of $\mathcal{S}$. By definition it is ${\rm In}(\mathcal{N})={\rm Out}(\mathcal{N})=0$, and also:
\begin{equation}
{\rm In}\left({\cal S}\right)={\rm Out}(\mathcal{S}^{c}),\,\,{\rm Out(\mathcal{S})={\rm In(\mathcal{S}^{c})}}\,. \label{eq:InOutN1}
\end{equation}

\begin{lemma}
\label{lem:equality} For any set ${\cal S}\subseteq{\cal N}$ , under any feasible allocation $\boldsymbol{d}\in\mathbb{D}$, it holds
\begin{equation}
\sum_{i\in\mathcal{S}}r_{i}+{\rm Out}\left({\cal S}\right)=\sum_{i\in\mathcal{S}}D_{i}+{\rm In}\left({\cal S}\right).\label{eq:equality}
\end{equation}
\begin{equation}
{\rm Out}\left({\cal S}\right)\leq\sum_{i\in{\cal S}}D_{i}\label{eq:OutIneq}
\end{equation}
 with equality holding iff all nodes in ${\cal S}$ give their resource only to nodes outside $\mathcal{S}$. Also, it is:
\begin{equation}
{\rm In}\left({\cal S}\right)\leq\sum_{i\in{\cal S}}r_{i}\label{eq:InIneq-1}
\end{equation}
with equality holding iff all nodes in $\mathcal{S}$ get resource only from nodes outside $\mathcal{S}$.
\end{lemma}
\begin{proof}
Note that for any node set ${\cal S}$, the following holds:
\begin{equation}
\sum_{i\in{\cal S}}\ \sum_{j\in{\cal N}_{i}\cap{\cal S}}d_{ji}=\sum_{i\in{\cal S}}\ \sum_{j\in{\cal N}_{i}\cap{\cal S}}d_{ij}.\label{eq:basic1}
\end{equation}
Also, by definition
\begin{equation}
r_{i}=\sum_{j\in{\cal N}_{i}\cap{\cal S}}d_{ji}+\sum_{j\in{\cal N}_{i}\cap{\cal S}^{c}}d_{ji},\label{eq:basic1-1}
\end{equation}
by feasibility of $\boldsymbol{d}\in\mathbb{D}$,
\begin{equation}
D_{i}=\sum_{j\in{\cal N}_{i}\cap{\cal S}}d_{ij}+\sum_{j\in{\cal N}_{i}\cap{\cal S}^{c}}d_{ij}.\label{eq:basic2}
\end{equation}
Hence we calculate
\begin{align*}
\sum_{i\in{\cal S}}r_{i} & =\sum_{i\in{\cal S}}\ \sum_{j\in{\cal N}_{i}\cap{\cal S}}d_{ji}+\sum_{i\in{\cal S}}\sum_{j\in{\cal N}_{i}\cap{\cal S}^{c}}d_{ji}\ {\rm by\ (}\ref{eq:basic1-1})\\
 & =\sum_{i\in{\cal S}}\ \sum_{j\in{\cal N}_{i}\cap{\cal S}}d_{ij}+{\rm In}\left({\cal S}\right)\ \ \ {\rm by\ (\ref{eq:basic1}),\, (\ref{eq:In-Out-3})}\\
 & =\sum_{i\in{\cal S}}\ \sum_{j\in{\cal N}_{i}\cap{\cal S}}d_{ij}+\sum_{i\in{\cal S}}\ \sum_{j\in{\cal N}_{i}\cap{\cal S}^{c}}d_{ij}+{\rm In}\left({\cal S}\right)-{\rm Out}\left({\cal S}\right)\\
 & =\sum_{i\in\mathcal{S}}D_{i}+{\rm In}\left({\cal S}\right)-{\rm Out}\left({\cal S}\right)\ \ \ {\rm by\ (\ref{eq:basic2})}
\end{align*}
Inequalities (\ref{eq:OutIneq}), (\ref{eq:InIneq-1}) follow directly
from the definitions. $\blacksquare$
\end{proof}
\begin{lemma}
\label{lem:levels} Let $\boldsymbol{d}\in\mathbb{D}$. If $K=1$
then $l_{1}=1$. If $K>1$ , then $l_{1}<1$ and $l_{K}>1$.
\end{lemma}
\begin{proof}
If $K=1$, we have $r_{i}=l_{1}D_{i}$ for all $i\in{\cal N}$. From
Lemma \ref{lem:equality} (applied with ${\cal S}\leftarrow{\cal N}$) we then have $l_{1}\sum_{i\in{\cal N}}D_{i}=\sum_{i\in{\cal N}}D_{i}$, hence $l_{1}=1$. Let now $K>1$. If $l_{1}\geq1$, then since $l_{K}>l_{1}$, we have
\begin{equation}
\sum_{i\in{\cal N}}r_{i} = \sum_{k=1}^{K}l_{k}\sum_{i\in{\cal \mathcal{L}}_{k}}D_{i}> \sum_{i\in{\cal N}}D_{i}
\end{equation}
which contradicts Lemma \ref{lem:equality}. Similarly is shown $l_{K}>1$. $\blacksquare$
\end{proof}
\begin{prop}
\label{prop:LexK1}Let $\bar{\boldsymbol{d}}\in\mathbb{D}.$ If $\bar{K}=1$
then $\bar{\boldsymbol{d}}$ is lex-optimal.
\end{prop}
\begin{proof}
From Lemma \ref{lem:levels}, $\bar{l}_{1}=1$ and hence $\bar{r}_{i}=D_{i},\ i\in{\cal N}$. If there is another allocation $\hat{\boldsymbol{d}}$ such $\hat{\boldsymbol{r}}\succ\bar{\boldsymbol{r}}$, then it should hold $\hat{r}_{i}\geq r_{i}=D_{i}\ \forall i\in{\cal N}$ and $\hat{r}_{j}>\bar{r}_{j}=D_{j}$ for at least one $j\in{\cal N}$. But then, it would be $\sum_{i\in{\cal S}}\hat{r}_{i}>\sum_{i\in{\cal S}}D_{i}$, which contradicts (\ref{eq:equality}) (applied with ${\cal S}\leftarrow{\cal N}$). $\blacksquare$

\end{proof}

Now we are ready to provide the \uline{proof of Theorem \ref{thm:MainTh0_K1}}: From Lemma \ref{lem:levels} we have that for a feasible allocation $\boldsymbol{d}\in\mathbb{D}$, with $K=1$, it is $l_1=1$. From Proposition \ref{prop:LexK1} we also get that this is a lex-optimal allocation. Moreover, since a lex-optimal allocation $\boldsymbol{d}^{*}$ is also feasible, when $K^{*}=1$, it is also $l_{1}^{*}=1$ from Lemma \ref{lem:levels}. $\blacksquare$

\vspace{1.5mm}

\uline{\textbf{3.2.2 PROOF of Theorem \ref{thm:MainTh0}}}. First, we need the following corollary.

\begin{cor}
\label{cor:GiveToK} If under a lex-optimal allocation $\boldsymbol{\bar{d}}$ it holds $k(\mathcal{\bar{D}}_{i})=K^{*}$ for some $i\in\mathcal{N}$, then ${\cal N}_{i}\subseteq\mathcal{L}_{K}^{*}$.\end{cor}
\begin{proof}
Since under a lex-optimal allocation there can be no node with level higher that $l_{K^{*}}$, Lemma \ref{lem:neighbor} is applied with the equality, i.e., $\forall\,j\in{\cal \bar{H}}_{i}$, it holds $l_{k(j)}^{*}=l_{k(\mathcal{\bar{D}}_{i})}^{*}=K^{*}$. Since the same also holds by definition for all nodes in $\mathcal{\bar{D}}_{i}$, the results follows. $\blacksquare$
\end{proof}

We introduce some additional definitions and results. Consider a lex-optimal allocation $\boldsymbol{\bar{d}}$ and let ${\cal \bar{Z}}$ be the subset of nodes in ${\cal L}_{K}^{*}$, with
the property: $i\in{\cal \bar{Z}}$ iff $\mathcal{L}_{k(\mathcal{\bar{D}}_{i})}^{*}=K^{*}$. Hence any node $i$ in ${\cal \bar{Z}}\subset{\cal L}_{K}^{*}$ gives resource only to nodes in ${\cal L}_{K}^{*}$. The next lemma shows that the set ${\cal \bar{Z}}$ is empty if $K^{*}\geq 2$.
\begin{lemma}
\label{lem:NoGive}Let $\bar{\boldsymbol{d}}$ be a lex-optimal allocation. If $K^{*}\geq2$, then ${\cal \bar{Z}}=\emptyset$, i.e., the nodes in ${\cal L}_{K}^{*}$ give all their resource to nodes outside ${\cal L}_{K}^{*}.$ Hence\footnote{\small{To facilitate the reader, we repeat the notation: $\overline{{\rm Out}}$ and ${\cal \bar{Z}}$ are annotated with the bar symbol since they depend on $\bar{\boldsymbol{d}}$, while the optimal level sets and the received resources are unique and hence annotated with the star symbol. }},
\begin{equation}
\overline{{\rm Out}}\left({\cal L}_{K}^{*}\right)=\sum_{i\in{\cal L}_{K}^{*}}D_{i}.\label{eq:OutZero}
\end{equation}
\end{lemma}
\begin{proof}
According to Corollary \ref{cor:GiveToK} all neighbors of any node $i\in{\cal \bar{Z}}$ are in ${\cal L}_{K}^{*}$. It follows that a node $i$ in ${\cal \bar{Z}}$ gets resource only from nodes in ${\cal \bar{Z}}$: if node $i$ was getting resource from a neighbor node $j\notin{\cal \bar{Z}}$, then since as the previous sentence says $j\in{\cal L}_{K}^{*}$, node $j$ should belong to ${\cal \bar{Z}}$ by definition; which is a contradiction. This implies that ${\rm \overline{In}}\left({\cal \bar{Z}}\right)=0$ and hence according to Lemma \ref{lem:equality}:
\[
\sum_{i\in{\cal \bar{Z}}}r_{i}^{*}\leq\sum_{i\in{\cal \bar{Z}}}D_{i}. \label{eq:xx}
\]
If ${\cal \bar{Z}}\neq\emptyset$, then since $r_{i}^{*}=l_{K}^{*}D_{i}$, $\forall i\in{\cal \bar{Z}},$ we conclude from (\ref{eq:xx}) that $l_{K}^{*}\leq1$, which contradicts Lemma \ref{lem:levels}. Equality (\ref{eq:OutZero}) follows immediately: since ${\cal \bar{Z}}=\emptyset,$ the nodes in ${\cal L}_{K}^{*}$ give all their resource to nodes outside ${\cal L}_{K}^{*}$ and hence (\ref{eq:OutIneq}) applies with equality. $\blacksquare$
\end{proof}
Let $\bar{{\cal G}}$ be the set of nodes from which nodes in ${\cal L}_{K}^{*}$ get resource, i.e., $\bar{{\cal G}}=(i\in\mathcal{N}: k(\bar{\mathcal{D}}_{i})=K^{*})$. It holds:
\begin{lemma}
\label{lem:PropOfG}
Let $\bar{\boldsymbol{d}}$ be a lex-optimal allocation and $K^{*}\geq2$. It holds \textup{${\cal L}_{K}^{*}\cap\bar{{\cal G}}=\emptyset$. Moreover,} the set $\bar{{\cal G}}$ is nonempty, independent, it holds ${\cal N}\left(\bar{{\cal G}}\right)={\cal L}_{K}^{*}$, and
\begin{equation}
{\rm \overline{In}}\left({\cal \bar{G}}\cup{\cal L}_{K}^{*}\right)=0.\label{eq:separation}
\end{equation}
\end{lemma}
\begin{proof}
According to Lemma \ref{lem:NoGive}, ${\cal L}_{K}^{*}\cap\bar{{\cal G}}={\cal \bar{Z}}=\emptyset$. Also, according to (\ref{eq:equality}), (\ref{eq:OutZero}), and the definition of $\bar{{\cal G}}$, it is
\[
\sum_{i\in{\cal L}_{K}^{*}}r_{i}^{*}=\overline{{\rm In}}\left({\cal L}_{K}^{*}\right)\leq\sum_{i\in{\cal \bar{G}}}D_{i}.
\]
Since $\sum_{i\in{\cal L}_{K}^{*}}r_{i}^{*}=l_{K}^{*}\sum_{i\in{\cal L}_{K}^{*}}D_{i}>0$, we get ${\cal \bar{G}}\neq\emptyset$. According to Corollary \ref{cor:GiveToK} and the definition of $\bar{{\cal G}},$ it holds ${\cal N}_{i}\subseteq{\cal L}_{K}^{*},\ \forall\,i\in\bar{{\cal G}}$. Since ${\cal L}_{K}^{*}\cap\bar{{\cal G}}=\emptyset$, $\bar{{\cal G}}$ is independent.

To show that ${\cal N}\left(\bar{{\cal G}}\right)={\cal L}_{K}^{*}$ we argue as follows. According to Corollary \ref{cor:GiveToK}, it is ${\cal N}\left(\bar{{\cal G}}\right)\subseteq{\cal L}_{K}^{*}$. Also, if ${\cal L}_{K}^{*}-{\cal N}\left(\bar{{\cal G}}\right)\neq\emptyset$, there would be a node $i\in{\cal L}_{K}^{*}$ not connected to any of the nodes in $\bar{{\cal G}}$; but since by definition of $\bar{{\cal G}}$ node $i$ gets resource only from nodes in $\bar{{\cal G},}$ we would then have $r_{i}^{*}=l_{K}^{*}D_{i}=0$, a contradiction since $l_{K}^{*}>1$ and $D_{i}>0$.

Notice next that ${\cal N}\left(\bar{{\cal G}}\right)={\cal L}_{K}^{*}$ and the set $\bar{{\cal G}}$ is independent, all neighbors of nodes in ${\cal \bar{G}}$ are in ${\cal L}_{K}^{*}$, and hence nodes in ${\cal \bar{G}}$ can get resource only from nodes in ${\cal L}_{K}^{*}$. Since by definition nodes ${\cal L}_{K}^{*}$ get resource only from ${\cal \bar{G}}$, (\ref{eq:separation}) holds. $\blacksquare$
\end{proof}
\begin{lemma}
\label{lem:ProdLess}
Let $\bar{\boldsymbol{d}}$ be a lex-optimal allocation and $K^{*}\geq2$. Let $k_{0}$ be the index of the smallest level in ${\cal \bar{G}}$$.$ Then $l_{K}^{*}l_{k_{0}}^{*}\leq1$. Strict inequality holds if
\begin{enumerate}
\item \textup{either ${\rm \overline{Out}}\left({\cal \bar{G}}\cup{\cal L}_{K}^{*}\right)>0,$}
\item \textup{or ${\rm \overline{Out}}\left({\cal \bar{G}}\cup{\cal L}_{K}^{*}\right)=0$
and ${\cal \bar{G}}-{\cal L}_{k_{0}}^{*}\neq\emptyset.$}
\end{enumerate}
If ${\rm \overline{Ou}t}\left({\cal \bar{G}}\cup{\cal L}_{K}^{*}\right)=0$ and ${\cal \bar{G}}_{K}-{\cal L}_{k_{0}}^{*}=\emptyset$, then $l_{K}^{*}l_{k_{0}}^{*}=1$.
\end{lemma}

\begin{proof}
Since by Lemma \ref{lem:PropOfG} ${\cal \bar{G}}$ is independent, and ${\cal N}\left(\bar{{\cal G}}\right)\subseteq{\cal L}_{K}^{*}$, the nodes in ${\cal \bar{G}}$ can give resource only to nodes in ${\cal L}_{K}^{*}$. Hence only nodes in ${\cal L}_{K}^{*}$ give resource to nodes in $\left({\cal \bar{G}}\cup{\cal L}_{K}^{*}\right)^{c}$, hence:
\begin{align*}
\overline{{\rm Out}}\left({\cal L}_{K}^{*}\right) & ={\rm \overline{In}}\left({\cal \bar{G}}\right)+{\rm \overline{Out}}\left({\cal \bar{G}}\cup{\cal L}_{K}^{*}\right).
\end{align*}
Taking into account (\ref{eq:OutZero}) we conclude:
\begin{equation}
\overline{{\rm In}}\left({\cal \bar{G}}\right)=\sum_{i\in{\cal L}_{K}^{*}}D_{i}-\overline{{\rm Out}}\left(\bar{{\cal G}}\cup{\cal L}_{K}^{*}\right).\label{eq:InIneq}
\end{equation}
Since nodes in ${\cal \bar{G}}$ constitute an independent set it follows:
\begin{equation}
{\rm \overline{Out}}\left({\cal \bar{G}}\right)=\sum_{i\in\bar{{\cal G}}}D_{i}.\label{eq:OutGk}
\end{equation}
From Lemma \ref{lem:equality} applied to set ${\cal \bar{G}}$, and using (\ref{eq:InIneq}-\ref{eq:OutGk}) we get
\begin{align}
&\sum_{i\in{\cal \bar{G}}}r_{i}=\overline{{\rm In}}\left({\cal \bar{G}}\right)=\sum_{i\in{\cal L}_{K}^{*}}D_{i}-\overline{{\rm Out}}\left(\bar{{\cal G}}\cup{\cal L}_{K}^{*}\right),\,\text{or} \nonumber \\
&\sum_{k=k_{0}}^{K^{*}-1}l_{k}^{*}\sum_{i\in{\cal L}_{k}^{*}\cap{\cal \bar{G}}}D_{i}=\sum_{i\in{\cal L}_{K}^{*}}D_{i}-\overline{{\rm Out}}\left(\bar{{\cal G}}\cup{\cal L}_{K}^{*}\right).\label{eq:rel1}
\end{align}
Next, since nodes in ${\cal L}_{K}^{*}$ get resource only from nodes in $\bar{{\cal G}}$ (and all of it) we have
\begin{equation}
l_{K}^{*}\sum_{i\in{\cal L}_{K}^{*}}D_{i}=\sum_{i\in\bar{{\cal G}}}D_{i}\,.\label{eq:rel2}
\end{equation}
Multiplying (\ref{eq:rel1}) and (\ref{eq:rel2}), and rearranging terms:
\begin{equation}
\sum_{k=k_{0}}^{K^{*}-1}l_{K}^{*}l_{k}^{*}\sum_{i\in{\cal L}_{k}^{*}\cap{\cal \bar{G}}}D_{i}=\sum_{i\in\bar{{\cal G}}}D_{i}-\frac{\sum_{i\in\bar{{\cal G}}}D_{i}}{\sum_{i\in{\cal L}_{K}^{*}}D_{i}}\overline{{\rm Out}}\left(\bar{{\cal G}}\cup{\cal L}_{K}^{*}\right).\label{eq:rel3}
\end{equation}
Eq. (\ref{eq:rel3}) implies that $l_{K}^{*}l_{k_{0}}^{*}\leq1$: if $l_{K}^{*}l_{k_{0}}^{*}>1$ then, because it will also hold $l_{K}^{*}l_{k}^{*}>1,\ \forall k\geq k_{0}$, (\ref{eq:rel3}) would not hold.

Now, if $\overline{{\rm Out}}\left(\bar{{\cal G}}\cup{\cal L}_{K}^{*}\right)>0$ then from (\ref{eq:rel3}) we have:
\[
\sum_{k=k_{0}}^{K^{*}-1}l_{K}^{*}l_{k}^{*}\sum_{i\in{\cal L}_{k}^{*}\cap{\cal \bar{G}}}D_{i}<\sum_{i\in\bar{{\cal G}}}D_{i}
\]
and arguing as above we see that necessarily $l_{K}^{*}l_{k_{0}}^{*}<1$. If $\overline{{\rm Out}}\left(\bar{{\cal G}}\cup{\cal L}_{K}^{*}\right)=0$ and ${\cal \bar{G}}-{\cal L}_{k_{0}}^{*}\neq\emptyset$ then again $l_{K}^{*}l_{k_{0}}^{*}<1$. To see this, notice that if $l_{K}^{*}l_{k_{0}}^{*}\geq1$ and ${\cal \bar{G}}-{\cal L}_{k_{0}}^{*}\neq\emptyset$ then it would hold:
\begin{equation}
l_{K}^{*}l_{k}^{*}\sum_{i\in{\cal L}_{k}^{*}\cap{\cal \bar{G}}}D_{i}\geq\sum_{i\in{\cal L}_{k}^{*}\cap{\cal \bar{G}}}D_{i}\ \forall\, k\geq k_{0}.\label{eq:HelpIneq}
\end{equation}
Also, since ${\cal \bar{G}}-{\cal L}_{k_{0}}^{*}\neq\emptyset$, for some $k>k_{0}$ there must be a nonempty set ${\cal L}_{k}^{*}\cap{\cal \bar{G}}$ which implies that the inequality is strict for some $k>k_{0}$. Adding inequalities (\ref{eq:HelpIneq}) we would then get,
\[
\sum_{k=k_{0}}^{K^{*}-1}l_{K}^{*}l_{k}^{*}\sum_{i\in{\cal L}_{k}^{*}\cap{\cal \bar{G}}}D_{i}>\sum_{i\in\bar{{\cal G}}}D_{i},
\]
which contradicts (\ref{eq:rel3}).

Assume finally that $\overline{{\rm Out}}\left(\bar{{\cal G}}\cup{\cal L}_{K}^{*}\right)=0$ and ${\cal \bar{G}}-{\cal L}_{k_{0}}^{*}=\emptyset$. From (\ref{eq:rel3})
we have $l_{K}^{*}l_{k_{0}}^{*}\sum_{i\in\bar{{\cal G}}}D_{i}=\sum_{i\in\bar{{\cal G}}}D_{i}$, and since $\bar{{\cal G}}\neq\emptyset$ implies $\sum_{i\in\bar{{\cal G}}}D_{i}>0$, we get $l_{K}^{*}l_{k_{0}}^{*}=1$. $\blacksquare$
\end{proof}
\begin{lemma}
\label{lem:Independent}If $K^{*}\geq2$, then the set ${\cal L}_{1}^{*}$
is independent.
\end{lemma}
\begin{proof}
Assume that for some pair $i,j\in{\cal L}_{1}^{*}$, it is $(i,j)\in{\cal E}$. Consider the largest set ${\cal C}\subseteq{\cal L}_{1}^{*}$ such that a) ${\cal C}$ contains both $i$ and $j$, b) the induced subgraph of ${\cal C}$ is connected. Therefore, each node in ${\cal C}$ has a node in ${\cal C}$, and hence a node in ${\cal L}_{1}^{*}$ as neighbor. By Lemma \ref{lem:neighbor} we have that under any lex-optimal allocation $\boldsymbol{d},$ it holds $\mathcal{L}_{k(\mathcal{D}_{i})}^{*}=1$ for all $i\in{\cal C}.$ That is, all nodes in ${\cal C}$ give resource only to other nodes in ${\cal C}$, hence, ${\rm Out}({\cal C})=0$. It follows from (\ref{eq:equality}) that $\sum_{i\in{\cal C}}r_{i}^{*}\geq\sum_{i\in{\cal C}}D_{i}$.
But we also have
\begin{eqnarray*}
\sum_{i\in{\cal C}}r_{i}^{*} & = & l_{1}^{*}\sum_{i\in{\cal C}}D_{i}< \sum_{i\in{\cal C}}D_{i},
\end{eqnarray*}
since $l_{1}^{*}<1$ by Lemma \ref{lem:levels}, which is a contradiction. $\blacksquare$
\end{proof}
\begin{lemma}
\textup{\label{lem:AllInG}}Let $\bar{\boldsymbol{d}}$ be a lex-optimal allocation and $K^{*}\geq2$.\textup{ It holds ${\cal L}_{1}^{*}\subseteq\bar{{\cal G}}$. }
\end{lemma}
\begin{proof}
Let $\bar{\mathcal{F}}_{k}\triangleq\left(\bar{{\cal G}}\cup{\cal L}_{K}^{*}\right)^{c}\cap{\cal L}_{k}^{*},\ k=1,...,K^{*}-1.$ It suffices to show that $\bar{\mathcal{F}}_{1}=\emptyset.$ Assume that $\bar{\mathcal{F}}_{1}\neq\emptyset$. Let ${\cal \bar{B}}$ be the set of nodes in $\left(\bar{{\cal G}}\cup{\cal L}_{K}^{*}\right)^{c}$ that are neighbors of nodes in $\bar{\mathcal{F}}_{1}$. The set ${\cal \bar{B}}$ is nonempty because otherwise, since ${\cal L}_{1}^{*}$ (and hence $\bar{\mathcal{F}}_{1}$ ) is an independent set and by Lemma \ref{lem:PropOfG} it is ${\cal N}\left(\bar{{\cal G}}\right)\subseteq{\cal L}_{K}^{*}$, all neighbors of any node in $\bar{\mathcal{F}}_{1}$ would be in ${\cal L}_{K}^{*}$ which implies that $\bar{\mathcal{F}}_{1}\subseteq\bar{{\cal G}}$, a contradiction.

Notice next that (\ref{eq:separation}) and Lemma \ref{lem:neighbor} imply that all nodes in ${\cal \bar{B}}$ give resource only to nodes in $\bar{\mathcal{F}}_{1}$ $.$ Hence
\[
l_{1}^{*}\sum_{i\in\mathcal{\bar{F}}_{1}}D_{i}\geq\sum_{i\in\mathcal{\bar{B}}}D_{i}.
\]

Also, since by Lemma \ref{lem:Independent} the set $\bar{\mathcal{F}}_{1}$ is independent, and by (\ref{eq:separation}) nodes in $\bar{\mathcal{F}}_{1}$ do not give resource to nodes in ${\cal L}_{K}^{*}$, all nodes in this set give resource only to nodes in ${\cal \bar{B}}$ and (notice that since ${\cal \bar{B}}\neq\emptyset$, is should hold $K^{*}-1\geq2$),
\[
\sum_{k=2}^{K^{*}-1}l_{k}^{*}\sum_{i\in\bar{\mathcal{F}}_{k}\cap{\cal \bar{B}}}D_{i}\geq\sum_{i\in\mathcal{\bar{F}}_{1}}D_{i}.
\]
 Multiplying the last two inequalities and canceling terms:
 \[
\sum_{k=2}^{K^{*}-1}l_{1}^{*}l_{k}^{*}\sum_{i\in\bar{\mathcal{F}}_{k}\cap{\cal \bar{B}}}D_{i}\geq\sum_{i\in\mathcal{\bar{B}}}D_{i},
\]
which implies that $l_{1}^{*}l_{K-1}^{*}\geq1$. But $l_{k_{0}}^{*}l_{K}^{*}>l_{1}^{*}l_{K-1}^{*},\ k_{0}\geq1$ and hence $l_{k_{0}}^{*}l_{K}^{*}>1$,
which contradicts Lemma \ref{lem:ProdLess}. $\blacksquare$
\end{proof}

Now we are ready to prove the following proposition.

\begin{prop}
\label{lem:MainLem}Let $\boldsymbol{d}^{*}$ be a lex-optimal allocation and $K^{*}\geq2$. The set ${\cal L}_{1}^{*}$ is independent, ${\cal N}\left({\cal L}_{1}^{*}\right)={\cal L}_{K}^{*}$, and
\begin{align}
&l_{1}^{*}l_{K}^{*}=1 \label{eq:MainLemProd1}\,,\\
&\sum_{i\in{\cal L}_{1}^{*}}r_{i}^{*}=\sum_{i\in{\cal L}_{K}^{*}}D_{i},\label{eq:MainLemL1LK}\\
&{\rm \overline{In}}\left({\cal L}_{1}^{*}\cup{\cal L}_{K}^{*}\right)={\rm \overline{Out}}\left({\cal L}_{1}^{*}\cup{\cal L}_{K}^{*}\right)=0,\label{eq:MainLemInOut0}
\end{align}
\end{prop}

\begin{proof}
By Lemma \ref{lem:neighbor} the nodes in the set ${\cal N}\left({\cal L}_{1}^{*}\right)=\cup_{i\in{\cal L}_{1}}{\cal N}_{i}-{\cal L}_{1}^{*}=\cup_{i\in{\cal L}_{1}}{\cal N}_{i}$ (the last equality hold because $\forall i\in{\cal L}_{i}^{*}$, it is ${\cal N}_{i}\cap{\cal L}_{1}^{*}=\emptyset$) give resource only to nodes in ${\cal L}_{1}^{*}$, hence
\begin{equation}
{\rm \overline{In}}\left({\cal L}_{1}^{*}\right)=\sum_{i\in{\cal N}\left({\cal L}_{1}^{*}\right)}D_{i} \,.\nonumber
\end{equation}
Also, since ${\cal L}_{1}^{*}$ is an independent set, its nodes give all their resource to nodes in ${\cal N}\left({\cal L}_{1}^{*}\right)$, hence it is:
\[
{\rm \overline{Out}}\left({\cal L}_{1}^{*}\right)=\sum_{i\in{\cal L}_{1}^{*}}D_{i}\,.
\]
Applying (\ref{eq:equality}) to the set ${\cal L}_{1}^{*},$ we then have
\begin{equation}
\sum_{i\in{\cal L}_{1}^{*}}r_{i}^{*}=l_{1}^{*}\sum_{i\in{\cal L}_{1}^{*}}D_{i}=\sum_{i\in{\cal N}\left({\cal L}_{1}^{*}\right)}D_{i}\label{eq:help2-1}
\end{equation}
On the other hand, since according to Lemma \ref{lem:NoGive} the nodes in ${\cal N}\left({\cal L}_{1}^{*}\right)\subseteq{\cal L}_{K}^{*}$ give all their resource to nodes outside ${\cal L}_{K}^{*}$, according to (\ref{eq:OutIneq}) applied with equality, we get
\begin{equation}
{\rm \overline{Out}}\left({\cal N}\left({\cal L}_{1}^{*}\right)\right)=\sum_{i\in{\cal N}\left({\cal L}_{1}^{*}\right)}D_{i}\,.\label{eq:d1}
\end{equation}
Moreover, ${\cal L}_{1}^{*}$ is an independent set and thus the nodes in ${\cal N}\left({\cal L}_{1}^{*}\right)$ get all the resource from nodes in ${\cal L}_{1}^{*}$. Hence:
\[
{\rm \overline{In}}\left({\cal L}_{K}^{*}\right)\geq\sum_{i\in{\cal N}\left({\cal L}_{1}^{*}\right)}D_{i}.
\]
Applying now (\ref{eq:equality}) to the set ${\cal N}\left({\cal L}_{1}^{*}\right)$ we have:
\begin{equation}
\sum_{i\in{\cal N}\left({\cal L}_{1}^{*}\right)}r_{i}^{*}=l_{K}^{*}\sum_{i\in{\cal N}\left({\cal L}_{1}^{*}\right)}D_{i}\geq\sum_{i\in{\cal L}_{1}^{*}}D_{i}.\label{eq:d2}
\end{equation}
Multiplying (\ref{eq:help2-1}), (\ref{eq:d2}) we get $l_{K}^{*}l_{1}^{*}\geq1$. If ${\rm \overline{Out}}\left({\cal \bar{G}}\cup{\cal L}_{K}^{*}\right)>0$, from Lemma \ref{lem:ProdLess} we have $l_{K}^{*}l_{1}^{*}<1$, i.e., a contradiction. If ${\rm \overline{Out}}\left({\cal \bar{G}}\cup{\cal L}_{K}^{*}\right)=0$ and ${\cal \bar{G}}-{\cal L}_{1}^{*}\neq\emptyset$ then from the same lemma we have $l_{K}^{*}l_{1}^{*}<1$, again a contradiction.

The only case that remains is ${\rm \overline{Out}}\left({\cal \bar{G}}\cup{\cal L}_{K}^{*}\right)=0$ and ${\cal \bar{G}}-{\cal L}_{1}^{*}=\emptyset$ (i.e., ${\cal \bar{G}}\subseteq{\cal L}_{1}^{*}$) which again by the lemma implies $l_{K}^{*}l_{1}^{*}=1.$ Also, Lemma \ref{lem:AllInG} implies ${\cal \bar{G}}={\cal L}_{1}^{*}.$
${\cal N}\left({\cal L}_{1}^{*}\right)={\cal L}_{K}^{*}$ follows from Lemma \ref{lem:PropOfG}, and (\ref{eq:MainLemL1LK}) follows from (\ref{eq:help2-1}). $\blacksquare$

\end{proof}

After providing this last proposition, we can \uline{proceed with the proof for Theorem \ref{thm:MainTh0}}: For $k=1,$ Items \ref{enu:MainTh0Item1}- \ref{enu:MainTh0Item3} follow from Proposition \ref{lem:MainLem}. Hence the theorem is true when $K\in\left\{ 2,3\right\} $. Assume now that $K\geq4$. Since according to Proposition \ref{lem:MainLem} it is ${\rm Out}\left({\cal L}_{1}\cup{\cal L}_{K}\right)={\rm In}\left({\cal L}_{1}\cup{\cal L}_{K}\right)=0$, the restriction of $\boldsymbol{d}$ in ${\cal Q}_{2}$, $\boldsymbol{d}_{{\cal Q}_{2}}=\left\{ d_{ij}\right\} _{(i,j)\in{\cal E}_{{\cal Q}_{2}}}$ is an allocation on the graph with $K-2$ levels. But then $\boldsymbol{d}_{{\cal Q}_{2}}$ must be a lex-optimal allocation in $G_{{\cal Q}_{2}}=({\cal Q}_{2},{\cal E}_{{\cal Q}_{2}})$ since otherwise we could combine an allocation $\hat{\boldsymbol{d}}_{{\cal Q}_{2}}\succ\boldsymbol{d}_{{\cal Q}_{2}}$ with the components of $\boldsymbol{d}$ in ${\cal E}-{\cal E}_{{\cal Q}_{2}}$ and get a lexicographically better allocation on the original graph. Moreover, by construction we have for the lowest level set in $G_{{\cal Q}_{2}}$: ${\cal L}_{{\cal Q}_{2},1}={\cal L}_{2}$ and ${\cal L}_{{\cal Q}_{2},K-2}={\cal L}_{K-1}={\cal L}_{K-2+1}$. Hence properties \ref{enu:MainTh0Item1}- \ref{enu:MainTh0Item3} hold for $k=2$ and we can repeat the same arguments for the graph $G_{{\cal Q}_{2}}$to deduce inductively the stated properties for all $k=1,....,\flr{K/2}$. $\blacksquare$

\section{Game-theoretic Analysis}\label{sec:game-theory-frameworks}
\subsection{Coalitional Game} \label{sec:coalitional-subsection}

We consider two notions for coalition stability \cite{myerson-gametheory-book}, namely weak and strong stability\footnote{\small{Please note that this service exchange game does not posses the \emph{comprehensive property}, due to the fact that nodes allocate their entire idle resource, and hence we cannot define the inner core and the Shapley values and compare them with our solution. For more details on this, please see \cite{myerson-gametheory-book}.}}. The latter is a more restrictive condition, and preferable as it ensures there is no other allocation that will yield a strictly better payoff \emph{even} for one user. The main result in this context is:
\vspace{-1mm}
\begin{theorem}
Any lex-optimal allocation $\bm{d}^{*}$ yields a received resource vector $\bm{r}^{*}$, that lies in the core of the NTU service exchange game, and it is strongly stable.
\label{thm:Stability}
\end{theorem}
\vspace{-1mm}
Therefore, no subset of nodes can deviate and improve the total received resource, \emph{for at least one} of its members, without reducing the total received resource of at least another one of its members. Combining Theorems 1 and 4 we have the following corollary:

\begin{cor}
\label{thm:MainLex}Let $K^{*}\geq2.$
Under any lex-optimal allocation $\bm{d}^{*}$, the respective received resource vector $\bm{r}^{*}$, belongs to the core of the NTU coalitional servicing game, and has the following structure:
\begin{enumerate}

\item The set of nodes ${\cal N}$ is partitioned into disjoint groups ${\cal M}_{1}^{*},...,{\cal M}_{L}^{*},$ where $\left\lceil K^{*}/2\right\rceil$, and each group contains nodes with exchange ratios belonging to at most two different levels.

\item There are exactly $\left\lfloor K^{*}/2\right\rfloor $ groups with 2 levels. For group ${\cal M}_{k}^{*},\ 1\leq k\leq\lfloor{K^{*}/2\rfloor}$ it holds ${\cal M}_{k}^{*}={\cal L}_{k}^{*}\cup{\cal L}_{K-k+1}^{*}.$

\item \label{enu:CorItems4}If $K^{*}$ is odd, there is also a group with one ratio level, i.e., $\mathcal{M}_{\lceil{K^{*}/2}\rceil}^{*}=\mathcal{L}_{\lceil{K^{*}/2}\rceil}^{*}$.

\item \label{enu:CorItem5}It holds, $l_{k}^{*}l_{K-k+1}^{*}=1,\ 1\leq k\leq\lfloor{K^{*}/2\rfloor}$ and if $K^{*}$ is odd, the single level group has ratio $l_{\lceil{K^{*}/2\rceil}}^{*}=1.$

\item The set $\cup_{k=1}^{\left\lfloor K^{*}/2\right\rfloor }{\cal {\cal L}}_{k}^{*}$ is independent.


\end{enumerate}
\end{cor}
\vspace{-1mm}
Finally, if $K^*=1$, there is one group of nodes with $l^*=1$. 

\textbf{Analysis and Discussion}. \uline{Existence of Core}. The above results reveal that this coalitional service exchange NTU game has always a non-empty core, for any graph $G$, and any resource endowments $\{D_i\}_{i\in\mathcal{N}}$. Moreover, the core contains all lex-optimal allocations, which are also strongly stable. This is a more demanding condition than the non-emptiness of the core.




\uline{Groups of Nodes}. Within the grand coalition, not all the nodes interact with each other. For example, in Figure \ref{fig:Corollary-Example} where $K^{*}=7$, all the lex-optimal allocations result in $4$ groups (denoted with the dotted circles): $\mathcal{M}_{1}^{*}=\mathcal{L}_{1}^{*}\cup\mathcal{L}_{7}^{*}$, $\mathcal{M}_{2}^{*}=\mathcal{L}_{2}^{*}\cup\mathcal{L}_{6}^{*}$, $\mathcal{M}_{3}^{*}=\mathcal{L}_{3}^{*}\cup\mathcal{L}_{5}^{*}$, and $\mathcal{M}_{4}^{*}=\mathcal{L}_{4}^{*}$. Each group consists of nodes belonging to one or two exchange ratio level sets, and none of them exchange resources with nodes in different groups\footnote{\small{Notice that these groups do not constitute coalitions according to the given definitions, and they are derived by the grand coalition solution.}}. These properties are very useful for network design. For example, for a given network we can predict which nodes will interact in the fair and stable allocation policy, and remove the redundant physical links, which in certain cases induce additional cost \cite{JacksonWolinsky1996}.

\subsection{Competitive Market}

In the competitive market framework, each node acts greedily, without any information about the graph or the other nodes' resources, and allocates its resource so as to maximize the total resource it receives in return. Interestingly, equilibriums always exist in this autonomous and decentralized setting, and lead to lex-optimal exchange ratio vectors:
\vspace{-1.5mm}
\begin{theorem}
\label{lem:reccip1}The following hold: \textbf{(i)} There is a lex-optimal allocation $\bm{d}^{*}$ under which each node $i\in\mathcal{N}$ gives resource to its neighbors in proportion to what it gets from them, i.e., $d_{ij}^{*}=d_{ji}^{*}D_{i}/{r_{i}^{*}}$, or $d_{ji}^{*}/d_{ij}^{*}=r_{i}^{*}/D_i=l_{k(i)}^{*},\, j\in \mathcal{N}_{i}$, and the neighbors not receiving resource from $i$ have higher exchange ratio, i.e., $l_{k(j)}^{*}\geq 1/l_{k(i)}^{*}=l_{k({\cal D}_{i})}^{*},\, j\in \mathcal{H}_{i}$. \textbf{(ii)} if the allocation satisfies the above conditions, then the allocation is lexicographically optimal.
\end{theorem}
\vspace{-1mm}
The proof of the theorem is provided in the Appendix.

\textbf{Analysis and Discussion}. This theorem states that there is a fair lex-optimal allocation, where every node $i\in\mathcal{N}$ serves its neighbors $j\in\mathcal{D}_i$ with a resource $d_{ij}$, so as to have a constant and equal exchange ratio $d_{ij}/d_{ji}$ with all of them. Therefore, the lex-optimal allocation is an exchange equilibrium, and, additionally, any possible exchange equilibrium is also a lex-optimal allocation. In other words, the competitive interactions of rational users embedded in a graph, lead to the same allocation point that a central designer would have selected for such a system.

\uline{Dynamic Model}. \rev{An important aspect to notice it that this framework can capture both models where infinitely divisible resources are exchanged among users with different preferences, and also dynamic settings where users exchange indivisible resources over time, exploiting their diverse resource availability. To make the latter case more clear, consider a dynamic resource exchange system which operates in the continuous time domain}. Every node $i\in\mathcal{N}$ creates service opportunities for its neighbors (or, \emph{tokens}\footnote{\small{These are 0-1 token allocation decisions: whenever a user has an idle resource, e.g., an amount of unused bandwidth or energy, it can allocate it to one of its neighbors.}}) according to a Poisson process with possibly different rate $\lambda_{i}>0$. A meaningful strategy from the perspective of the nodes is the following: each node $i$ allocates a token generated at time $t$ to its neighbor that has, until then, given to $i$ the largest number of service tokens (per received token from $i$).

A rational user, with no information about the graph and the nodes' endowments, is reasonable to expect that this strategy can increase its benefit. Besides, this type of best response policies have been considered before, e.g., for P2P networks \cite{zhang-proportional} where it was shown that they converge to a steady state. In our case however, the scheme is decentralized and totally asynchronous. Interestingly, numerical results in Sec. \ref{sec:Numerical-Results} indicate that such myopic policies do converge to a steady state which moreover satisfies Theorem \ref{lem:reccip1}. This means that this dynamic model has a steady state that asymptotically coincides with the respective static model, and the equilibrium can be found if we set $D_i=\lambda_i$, $\forall\,i\in\mathcal{N}$. This is very important as it reveals that the results of this work do not apply only for the above static models, but also characterize the steady state allocations and equilibriums of more dynamic systems, where idle resources or service opportunities are created and allocated by each node asynchronously, greedily, and with no global network information (i.e., beyond the one hop neighbors).

\section{Lex-optimal Algorithms}\label{sec:algorithms}

In this section, we provide a polynomial (in $|\mathcal{N}|$) time algorithm that finds the lex-optimal allocation, and the respective exchange ratio vectors. The proposed algorithm uses the idea of max-min programming algorithm  proposed in \cite{boudecfairness07} and takes advantage of the structure of lex-optimal exchange ratio vector described in Theorem \ref{thm:MainTh0} to improve performance. The algorithm is based on the following result.
\begin{lemma}\label{lem:MaxMin}
Let $\bar{\bm{d}}\in\mathbb{D}$ and $K\geq2$. If the set $\bar{{\cal L}}_{1}$ is independent and
\begin{equation}
\sum_{i\in\bar{{\cal L}}_{1}}\bar{r}_{i}=\sum_{i\in{\cal N}\left(\bar{{\cal L}}_{1}\right)}D_{i},\label{eq:AllFromD}
\end{equation}
then:
\begin{enumerate}
\item \label{enu:ItemH1} For any allocation $\hat{\bm{d}}$ that solves the problem
\begin{equation}
\underset{\bm{d}\in\mathbb{D}}{{\rm maximize}}\min_{j\in{\cal N}}\frac{r_{j}}{d_{j}}\label{eq:maxmmin}
\end{equation}
it holds $\bar{{\cal L}}_{1}\subseteq\widehat{{\cal L}}_{1}$ and $\bar{l}_{1}=\hat{l}_{1}$.
\item \label{enu:ItemH2} The set $\bar{\mathcal{L}}_{1}$ coincides with the respective set of the lex-optimal ratio vector, i.e., $\bar{{\cal L}}_{1}={\cal L}_{1}^{*}$ and $\bar{l}_{1}=l_{1}^{*}$.
\end{enumerate}
\end{lemma}
\begin{proof}
$1)\,$ Since $\hat{\boldsymbol{d}}$ solves (\ref{eq:maxmmin}) it holds $\hat{l}_{1}\geq\bar{l}_{1}$ and hence $\frac{\hat{r}_{i}}{D_{i}}\geq\hat{l}_{1}\geq\bar{l}_{1}=\frac{\bar{r}_{i}}{D_{i}},\ i\in\bar{{\cal L}_{1}}$
i.e.,
\begin{equation}
\hat{r}_{i}\geq\bar{r}_{i},\ \ i\in\bar{{\cal L}}_{1}.\label{eq:ineqrrate}
\end{equation}
We will show next that equality holds in (\ref{eq:ineqrrate}) which implies that $\bar{{\cal L}}_{1}\subseteq\widehat{{\cal L}}_{1}$. To see this notice that if strict inequality holds for at least one $i\in\bar{{\cal L}_{1}}$ then
\begin{equation}
\sum_{i{\cal \in\bar{L}}_{1}}\hat{r}_{i}>\sum_{i\in\bar{{\cal L}}_{1}}\bar{r}_{i}.\label{eq:RateIneq-1}
\end{equation}
But since $\bar{{\cal L}}_{1}$ is an independent set, we have
\begin{equation}
\sum_{i{\cal \in\bar{L}}_{1}}\hat{r}_{i}={\rm \widehat{In}}\left(\bar{{\cal L}}_{1}\right) \leq\sum_{i\in{\cal N}\left(\bar{{\cal L}}_{1}\right)}D_{i}=\sum_{i\in\bar{{\cal L}}_{1}}\bar{r}_{i}
\end{equation}
where the inequality holds by definition of inflow, and the last equality by assumption. This result contradicts (\ref{eq:RateIneq-1}).

$2)\,$ Since any lex-optimal allocation $\hat{\boldsymbol{d}}$ solves (\ref{eq:maxmmin}), we have $\bar{{\cal L}}_{1}\subseteq\hat{{\cal L}}_{1}={\cal L}_{1}^{*}.$
If $\bar{{\cal L}}_{1}$ were a strict subset of $\hat{{\cal L}_{1}}$, then $\bar{\boldsymbol{d}}$ would be lexicographically better that $\hat{\boldsymbol{d}},$ a contradiction. $\blacksquare$
\end{proof}

Algorithm $1$ provides the details. Recall the definition of graphs $G_{{\cal Q}_{k}}=\left({\cal Q}_{k},{\cal E}_{{\cal Q}_{k}}\right)$ used in Theorem \ref{thm:MainTh0}. The number of iterations of the algorithm is equal to the number of sets ${\cal Q}_{k}$, i.e., $\lceil{K^{*}/2}\rceil$. Since each of the level sets contains at least 2 nodes, it holds $K^{*}\leq N/2$, and hence the number of iterations is at most $\lceil{N/4}\rceil$.

In Step \ref{enu:solve_max_min}, $l_{k}^{*}$ is computed as the optimal value of optimization problem (1.1.), as ensured by Lemma \ref{lem:MaxMin}.
The latter optimization problem can be transformed to a linear programming problem and hence can be solved in polynomial time. Note that the dimensionality
of the problem is reduced at each iteration. If the conditions in Steps \ref{enu:if--stop} and \ref{enu:is-stop1} are satisfied, then the lex-optimal allocation
has been determined on all links and the algorithm terminates.

The implementation and polynomial complexity of Step \ref{enu:Find_L1} will be discussed shortly. This step determines the set ${\cal L}_{k}^{*}$,
and hence ${\cal L}_{K^{*}-k+1}^{*}={\cal N}\left({\cal L}_{k}^{*}\right)$ and $l_{K^{*}-k+1}^{*}=1/l_{k}^{*}$. Also, at the exit from this step,
the allocated resources of all outgoing links from nodes in ${\cal L}_{K^{*}}^{*}$ to nodes in ${\cal N}-\left({\cal L}_{K^{*}}^{*}\cup{\cal L}_{1}^{*}\right)$,
will be zero and the allocated resources of all outgoing links from nodes in ${\cal L}_{K^{*}}^{*}$ to nodes in ${\cal L}_{1}^{*}$ will be determined.

\begin{algorithm}
\nl $k\leftarrow1$; \\%
\nl \While{1}{
\nl \label{enu:solve_max_min} Find $\bm{\hat{r}}$ and $\hat{\bm{d}}$ solving:
$\underset{\bm{d}\in\mathbb{D}_{{\cal Q}_{k}}}{{\rm maximize}}\min_{j\in{\cal N}_{{\cal Q}_{k}}}\frac{r_{j}}{D_{j}};$ (1.1) \\
\nl Set $l_{{\cal Q}_{k},1}^{*}$ to the value of the solution to (1.1.) \\
\nl \label{enu:if--stop}\textbf{If} $\left(\hat{r}_{i}/D_{i}=1\ i\in{\cal Q}_{k}\right)\,$ \textbf{then} $\,K^{*}=k$; \uline{Exit};  \\%
\nl \label{enu:Find_L1}Find the set ${\cal L}_{{\cal Q}_{k},1}^{*}$; \\ %
\nl Determine set ${\cal L}_{{\cal Q}_{k},2}^{*}={\cal N}\left({\cal L}_{{\cal Q}_{k},1}^{*}\right)$
and level value $l_{{\cal Q}_{k},2}^{*}=1/l_{{\cal Q}_{k},1}^{*}$ \\%
\nl \label{enu:ZeroAlloc} $d_{ji}\leftarrow 0$, $\forall$  $(j,i)$ $i\in{\cal L}_{{\cal Q}_{k},2}^{*}$, $j\in{\cal Q}_{{\cal Q}_{k}}$ - $\left({\cal L}_{{\cal Q}_{k},1}^{*}\cup{\cal L}_{{\cal Q}_{k},2}^{*}\right)$; \\%
\nl \label{enu:FindRates} Find lex-optimal allocations on links $\left(i,j\right),\ i\in{\cal L}_{{\cal Q}_{k},1}^{*},\ j\in{\cal L}_{{\cal Q}_{k},2}^{*}$; \\%
\nl \label{enu:is-stop1} \textbf{If} $\left({\cal L}_{{\cal Q}_{k},1}^{*}\cup{\cal L}_{{\cal Q}_{k},2}^{*}={\cal Q}_{k}\right)\,$ \textbf{then} {$\,K^{*}=k;$ \uline{Exit};}\\%
\nl $k\leftarrow k+1$; }%
\caption{Finding the Lex-optimal allocation}
\end{algorithm}

Step \ref{enu:ZeroAlloc} sets to zero all allocations of incoming links from nodes in ${\cal Q}_{k}-\left({\cal L}_{{\cal Q}_{k},1}^{*}\cup{\cal L}_{{\cal Q}_{k},2}^{*}\right)$
to nodes in ${\cal L}_{{\cal Q}_{k},2}^{*}$, as is required by Theorem \ref{thm:MainTh0}. Step \ref{enu:FindRates} determines allocations $d_{ij}^{*},\ i\in{\cal L}_{k}^{*},\ j\in{\cal L}_{K^{*}-k+1}^{*}$. Since it is known by Theorem \ref{thm:MainTh0} that
\[
\sum_{j\in{\cal N}_{i}}d_{ji}^{*}=1/l_{k}^{*},\ i\in{\cal L}_{K^{*}-k+1}^{*},
\]
this steps is equivalent to finding a feasible solution to a linear programming problem and hence it takes polynomial time to execute.

\begin{algorithm}
\nl ${\cal L}=\widehat{{\cal L}}_{1}$; $r_{ij}=\widehat{r}_{ij}$; /{*} \emph{on exit ${\cal L}={\cal L}_{1}^{*}*/$} \\%
\nl \label{enu:While1}\While{$\exists\,(i,j)$ where $i,j\in{\cal L}$
\textbf{and $d_{ij_{1}}>0$ }for some $j_{1}\in{\cal N}\left({\cal L}\right)$}{
\nl \label{enu:Reallocate-rate}Reallocate resource from link $(i,j_{1})$ to link $(i,j)$ ensuring that with the new allocation $\min\{r_{j}/D_{i},\ r_{j_{1}}/D_{i}\}>l_{1}^{*}$; \\%
\nl Set ${\cal L}\leftarrow{\cal L}-\{j\}$; } /{*}\emph{on exit the set ${\cal L}$ is independent}{*}/ \\
\nl \label{enu:While2} \While{$\exists\,(i,j_{1})$, $i\in{\cal N}\left({\cal L}\right),\ j_{1}\in{\cal N}-{\cal L}$ with $r_{ij_{1}}>0,$}{
\nl Reallocate resource from node $j_{1}$ to a node $j$ in ${\cal N}\left(i\right)$ ensuring that with the new allocation $\min\{r_{j}/D_{i},\ r_{j_{1}}/D_{i}\}>l_{1}^{*}$; \\%
\nl Set ${\cal L}\leftarrow{\cal L}-\{j\}$; \emph{/{*} on exit, set ${\cal L}$ satisfies (\ref{eq:AllFromD}) {*}/}
\caption{Finding the set $\mathcal{L}_{1}^{*}$}}

\end{algorithm}

It remains to show that Step \ref{enu:Find_L1} has polynomial complexity. According to Lemma \ref{lem:MaxMin}, the solution
to (\ref{eq:maxmmin}) determines $l_{1}^{*}<1$ and in general provides a solution $\widehat{{\cal L}}_{1}$ which is a superset of ${\cal L}_{1}^{*}$.
Furthermore, if by reallocating some of the link resources $\widehat{d}_{ij}$ we are able to create an allocation $\bar{\bm{d}}$ such that (i) the
set $\bar{{\cal L}_{1}}$ is independent, and (ii) the relation $\sum_{i\in\bar{{\cal L}}_{1}}\bar{r}_{i}=\sum_{i\in{\cal N}\left(\bar{{\cal L}}_{1}\right)}D_{i}$ holds,
then it will be ${\cal L}_{1}^{*}=\bar{{\cal L}}_{1}$.

The resource reallocation is described in Algorithm 2. There are two iteration loops. First, starting from the set ${\cal L}=\widehat{{\cal L}}_{1}$, if there
is a link $(i,j)$ such that $i,j\in{\cal L}$ then we select a node $j_{1}\in{\cal N}\left({\cal L}\right)$, with $d_{ij_{1}}>0$, and  transfer resource from link $(i,j_{1})$
to the link $(i,j)$. This selection is always possible since otherwise the condition ``for all links $(i,j)$ such that $i,j\in{\cal L}$ there is no node
$j_{1}\in{\cal N}\left({\cal L}\right)$, with $d_{ij_{1}}>0$'' would hold; however, this implies that $l_{1}^{*}=1$ which is excluded because at this point
we have $K^{*}\geq2$. The transfer of resource from$(i,j_{1})$ to link $(i,j)$ in Step \ref{enu:Reallocate-rate} of the algorithm ensures that the received resource
ratios of nodes $j$ and $j_{1}$ are larger than $l_{1}^{*}$ and hence $j$ necessarily does not belong to ${\cal L}_{1}^{*}.$ Hence on exit from the while loop
in Step \ref{enu:While1} the set ${\cal L}$ is independent. However, in order to ensure equality to ${\cal L}_{1}^{*}$ we may need to further modify ${\cal L}$
to ensure that the condition (\ref{eq:AllFromD}) holds. This is done in the second while loop that starts at Step \ref{enu:While2}. Also, at the exit from the
algorithm, as a result of this reallocation process, the allocated resources of all outgoing links from nodes
in ${\cal L}_{K^{*}}^{*}$ to nodes in ${\cal N}-\left({\cal L}_{K^{*}}^{*}\cup{\cal L}_{1}^{*}\right)$ will be zero and, the allocated resources of all
outgoing links from nodes in ${\cal L}_{K^{*}}^{*}$ to nodes in ${\cal L}_{1}^{*}$ will be determined. As is clear from the above description, Algorithm
$1$ and $2$ take polynomial time to execute.

\section{Numerical Examples}\label{sec:Numerical-Results}

In this section, we analyze representative numerical examples to shed light on the above results. Consider first the network of Fig. $\ref{fig:1st-example-6nodes}$ which has $6$ nodes. Solid lines represent the physical connections of each node and dotted arrows indicate resource allocation. Next to each node $i$ we depict its resource endowment. At the lex-optimal point, we have $K^{*}=3$ levels with 3 node sets $\mathcal{L}_{1}^{*}={\{1,6\}},\,\mathcal{L}_{2}^{*}={\{3,4\}},\;\mathcal{L}_{3}^{*}={\{2,5\}}$ which are marked with different colors.

\begin{figure}
\centering
\includegraphics[scale=0.47]{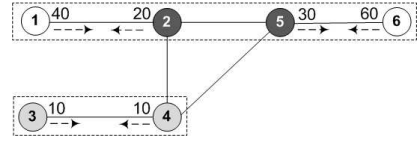}
\caption{\small{A network with 6 nodes that create 2 groups, each one marked with the dotted-line rectangle. There are 3 different levels of exchange ratios. The color of each node is analogous to its exchange ratio value (increasing from white to black colour). The received resources are $r_{1}^{*}=20,$ $r_{2}^{*}=40,$ $r_{3}^{*}=10,$ $r_{4}^{*}=10$, $r_{5}^{*}=60,$ and $r_{6}^{*}=30$. The exchange ratios for the nodes belonging to each set are $l_{1}^{*}=0.5$, $l_{2}^{*}=1$, and $l_{3}^{*}=2$ respectively.}}\label{fig:1st-example-6nodes}
\end{figure}

Let us now verify the properties that the lex-optimal allocation should have according to Theorem $\ref{thm:MainTh0}$. First, notice that set $\mathcal{L}_{1}^{*}$ is independent in graph $G$. Moreover, all the neighbors of nodes in set $\mathcal{L}_{3}^{*}$, i.e., nodes $2$ and $5$, belong in $\mathcal{L}_{1}^{*}$. Although nodes in $\mathcal{L}_{3}^{*}$ are physically connected, they only allocate resource to nodes in $\mathcal{L}_{1}^{*}$ and it holds $\sum_{i\in\mathcal{L}_{3}^{*}}D_{i}=\sum_{i\in l_{1}}r_{i}=20+30$. Moreover, the highest and the lowest levels satisfy the condition $l_{1}^{*}l_{3}^{*}=1$.

Similarly, we can verify that the structure of the lex-optimal allocation satisfies Corollary $\ref{thm:MainLex}$. The nodes are partitioned into $2$ disjoint groups $\mathcal{M}_{1}^{*}=\mathcal{L}_{1}^{*}\cup\mathcal{L}_{3}^{*}$ and $\mathcal{M}_{2}^{*}=\mathcal{L}_{2}^{*}$, each one containing nodes with at most two levels. Also, the nodes in $\mathcal{L}_{1}^{*}$ are connected only to nodes in the set $\mathcal{L}_{3}^{*}$, and it is $l_{2}^{*}=1$. Finally, the conditions of Theorem $\ref{lem:reccip1}$ are satisfied. For example, node $2$ allocates resource only to node $1$, with $l_{k(1)}=1/l_{k(2)}$, and not to node $4$ since it is $l_{k(4)}=1>0.5=l_{k(\mathcal{D}_{2})}$, where $l_{k(\mathcal{D}_{2})}=l_{k(1)}$.

For the example of Fig. $\ref{fig:3rd-example-13nodes}$ we used a network with $13$ nodes that yields $K^{*}=6$ levels, with $l_{1}^{*}=0.25$, $l_{2}^{*}=0.43$, $l_{3}^{*}=0.77$, $l_{4}^{*}=2.34$, $l_{5}^{*}=1.3$, and $l_{6}^{*}=4$. The sets are $\mathcal{L}_{1}^{*}=\{12,\,13\}$, $\mathcal{L}_{2}^{*}=\{4,\,6,\,8,\,10\}$, $\mathcal{L}_{3}^{*}=\{2\}$, $\mathcal{L}_{4}^{*}=\{1\}$, $\mathcal{L}_{5}^{*}=\{3,\,5,\,7,\,9\}$, and $\mathcal{L}_{6}^{*}=\{11\}$. Sets $\mathcal{L}_{1}^{*}$, $\mathcal{L}_{2}^{*}$, and $\mathcal{L}_{3}^{*}$ are independent in graphs $G_{\mathcal{Q}_{1}},$ $G_{\mathcal{Q}_{2}}$, $G_{\mathcal{Q}_{3}}$, and the set $\mathcal{L}_{1}^{*}\cup\mathcal{L}_{2}^{*}\cup\mathcal{L}_{3}^{*}$ is independent in $G$. Moreover, it is $\mathcal{L}_{6}^{*}=\mathcal{N}_{\mathcal{Q}_{1}}(\mathcal{L}_{1}^{*})$, $\mathcal{L}_{5}^{*}=\mathcal{N}_{\mathcal{Q}_{2}}(\mathcal{L}_{2}^{*})$ and $\mathcal{L}_{4}^{*}=\mathcal{N}_{\mathcal{Q}_{3}}(\mathcal{L}_{3}^{*})$, and holds $l_{6}^{*}l_{1}^{*}=l_{5}^{*}l_{2}^{*}=l_{4}^{*}l_{3}^{*}=1$. In this example we have 3 disjoint groups $\mathcal{M}_{1}^{*}=\mathcal{L}_{1}^{*}\cup\mathcal{L}_{6}^{*}$, $\mathcal{M}_{2}^{*}=\mathcal{L}_{2}^{*}\cup\mathcal{L}_{5}^{*}$, and $\mathcal{M}_{3}^{*}=\mathcal{L}_{3}^{*}\cup\mathcal{L}_{4}^{*}$. We see that links $(10,11)$, $(5,11)$, $(1,3)$, $(1,5)$ and $(2,7)$ are redundant and can be removed without affecting the lex-optimal allocation.

\begin{figure}
\centering
\includegraphics[scale=0.47]{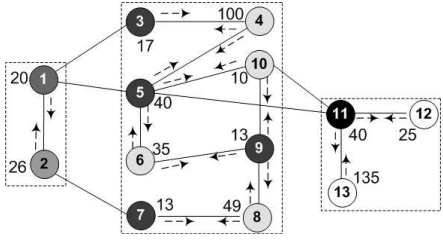}
\caption{\small{A network with $13$ nodes which create $3$ groups. Received resources are $r_{1}^{*}=26$, $r_{2}^{*}=20$, $r_{3}^{*}=39.74$, $r_{4}^{*}=42.78$, $r_{5}^{*}=93.49$, $r_{6}^{*}=14.97$, $r_{7}^{*}=30.38$, $r_{8}^{*}=20.96$, $r_{9}^{*}=30.38$, $r_{10}^{*}=4.28$, $r_{11}^{*}=160$, $r_{12}^{*}=6.25$, and $r_{13}^{*}=33.75$. \label{fig:3rd-example-13nodes}}}
\end{figure}

\begin{figure}
\centering
\includegraphics[scale=0.47]{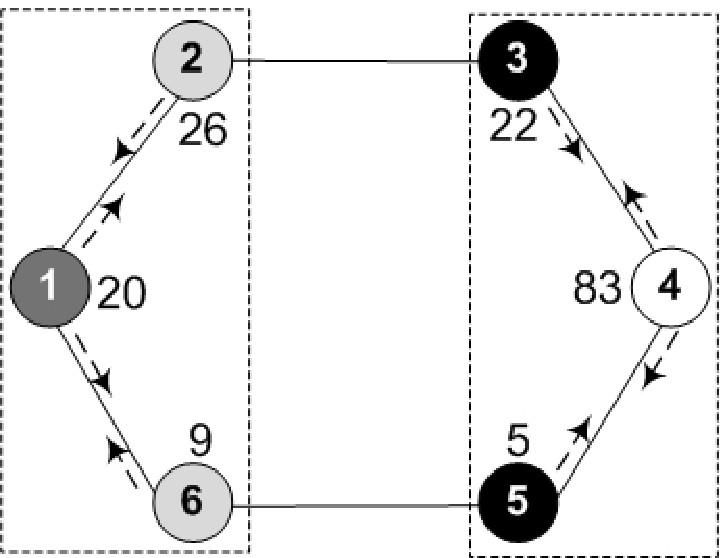}$\qquad$\includegraphics[scale=0.45]{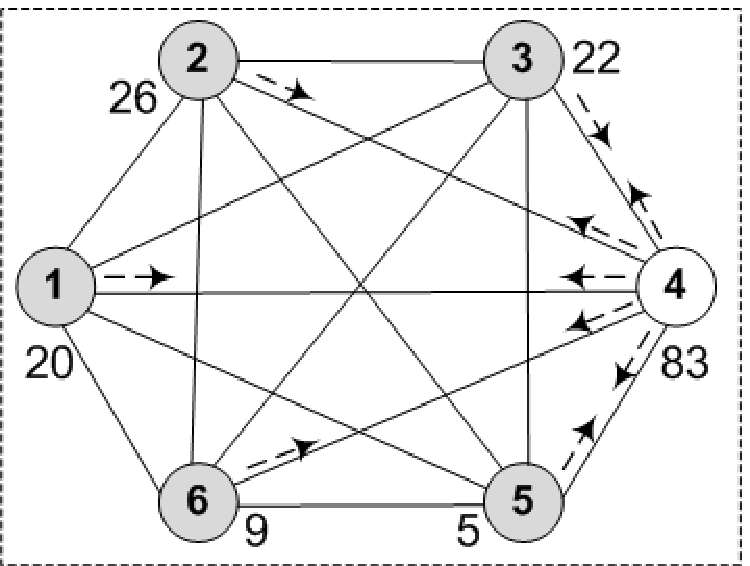}
\caption{\small{A complete graph of $6$ nodes with $1$ coalition and $2$ levels.}}\label{fig:complete-graph-6nodes}
\end{figure}

Figure $\ref{fig:complete-graph-6nodes}$ (right) depicts a complete graph with $6$ nodes, where node $i=4$ has level $l_{1}^{*}=0.988$ while the other nodes have level $l_{2}^{*}=1.012$. In general for complete graphs, from Property 6 of Corollary $\ref{thm:MainLex}$ and the fact that independent sets in such graphs contain only one node, it follows that lex-optimal allocations may have at most two levels. Moreover a complete graph has two levels iff the resource of node $i_{0}$ with the maximum endowment is larger than the sum of the resources of the rest of the nodes, and it is $\mathcal{L}_{1}=\{i_{0}\}$. On the other hand, for the respective $6$-node ring graph, the lex-optimal solution yields $2$ groups and $4$ levels.

\begin{figure}
\vspace{-1mm}
\centering
\includegraphics[scale=0.42]{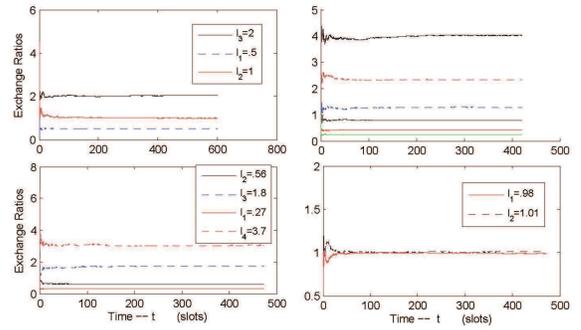}
\caption{\small{Convergence results for the dynamic and asynchronous best response strategy of nodes, for networks in Fig. \ref{fig:1st-example-6nodes} (upper left), Fig. \ref{fig:3rd-example-13nodes} (upper right), Fig. \ref{fig:complete-graph-6nodes}-ring graph (lower left), and Fig. \ref{fig:complete-graph-6nodes}-complete graph (lower right). Horizontal axis is time, and vertical axis shows the exchange ratio level values.}}
\label{figure:convergence}
\vspace{-1mm}
\end{figure}

Finally, we show that the naive best response strategy of the nodes in any graph-constrained dynamic resource exchange market, converges to a steady state point. Moreover, the latter coincides with the lex-optimal point of a static market in which every node has an average resource that is equal to the respective token generation rate of the dynamic market, i.e., $\lambda_i=D_i,\,\forall\,i\in\mathcal{N}$. In Figure \ref{figure:convergence} we present the quite fast convergence (each slot corresponds to the creation of a service opportunity) of this scheme for the above four networks, where we see that the system converges to the expected ratio values.

\section{Related Works}\label{sec:Related}

The model we consider is generic and representative for many communication or economic networks. For example, such models arise in graphical economies \cite{KearnsGraphEcon2004}, \cite{KearnsEconSocial2004}, which extend the classical Walrasian equilibrium \cite{ColellWhinstonGreenBook1995} and Arrow - Debreu analysis \cite{arrow-debreu} by imposing graph constraints on the subsets of buyers and sellers that can trade. However, our model does not presume any type of money transfers, i.e., there is no budget constraints (as in typical exchange economies) and the nodes do not value money (as in market games) \cite{osborne}. \rev{Similar bartering models have been studied for housing markets \cite{shapley-scarf} or timeshare exchanges \cite{krishna06}, where the focus has been to prove existence of equilibriums.}

Here, we fully characterize the equilibriums, relate them to the max-min fair solution, and study how they are affected by the graph. We also prove that these exchange equilibriums lie within the core of the respective NTU game. Although this relation is known for market games and the respective coalitional games \cite{osborne}, to the best of our knowledge this is the first result for NTU coalitional graph-constrained games without money. This property is also related to \emph{strong} Nash equilibriums (see \cite{strong-nash} and references therein), for which however there are no general existence results, nor they are appropriate for this competitive framework. Finally, \cite{Herings2000}, and \cite{JacksonWolinsky1996} studied also core solutions of coalitional graph games where the nodes are allowed to create new or severe existing connections. In our model the graph is exogenously given, e.g., based on the location of the nodes.


The problem of enabling cooperation in networks (or, networked systems) is of paramount importance and has been considered in different contexts, such as routing in ad hoc networks \cite{hubaux-coop}, WiFi sharing models \cite{efstathiou}, mesh networks \cite{wu-mesh}, and P2P overlays \cite{RJohariToNBilateral2011}. This is a problem that gains increasing interest in communication networks \cite{Sofia-UPN}, \cite{bewifi}, \cite{FON}, \cite{confine}, and in social and economic networks as well \cite{collaborativeconsumption}, \cite{KearnsGraphEcon2004}, \cite{JacksonWolinsky1996}. Unlike previous works, our model does not presume any kind of infrastructure, e.g., for transaction or reputation systems. Instead, we show numerically that asynchronous best response algorithms, with no information about the graph and resource endowments, converge to a fair and robust (i.e., in the core) exchange equilibrium.

Previous works e.g., \cite{zhang-proportional} have studied similar mechanisms for P2P file sharing systems, without however characterizing its properties and relation to competitive and coalitional equilibriums. The max-min fair criterion is natural for this setting, as it is defined with respect to each user's contribution. Also, while our model is similar to previous works, e.g., see \cite{RJohariToNBilateral2011} and references therein, our analysis provides novel insights for the structure and properties of the resulting equilibriums, and we also propose polynomial-time algorithms for their calculations. These algorithms can be also used for deriving the competitive equilibriums in graphical economies \cite{KearnsEconSocial2004}, \cite{KearnsGraphEcon2004}. 

\section{Discussion and Conclusions}\label{sec:Conclusions}

We considered a service (or, resource) exchange model among self-interested nodes embedded in a graph that prescribes their possible interactions. This is a key network model that represents Internet sharing communities \cite{FON}, \cite{Sofia-UPN}, \cite{bewifi}, \cite{confine}, P2P file sharing \cite{RJohariToNBilateral2011}, \cite{ioannidis-peer-assisted}, energy sharing networks \cite{saad-smart}, graphical economies \cite{KearnsGraphEcon2004}, \cite{KearnsEconSocial2004} and many online resource sharing platforms \cite{collaborativeconsumption}, \cite{neighborgood}, \cite{adalbdal}, \cite{swapit}, \cite{homeexchange}. \rev{Such systems can be dynamic where the users share their resource surpluses that they have in different (and diverse) time instances, or static where users having different resource preferences barter with each other so as to acquire the resources they value higher}. Despite the large interest of the research community and the previous contributions for specific related models (e.g., for P2P overlays), the fundamental properties of these systems remain unexplored.

We showed that the max-min fair policy exhibits a very rich structure, and characterized its properties for any given graph and node resource endowments. More importantly, we proved that this policy coincides with the exchange equilibrium of the respective competitive game, and lies in the core of the respective NTU coalitional game. This important result reveals that there is a unifying approach that solves the resource allocation problem for graph-constrained systems (or, economies), for different node behaviors. In other words, we can apply the max-min fair criterion, that has been extensively used for load balancing in centralized communication networks (e.g., see \cite{nace-tutorial} and references therein), to service exchange models with autonomous and selfish nodes.

Finally, our findings contribute to the game theoretic literature since the connection between the competitive equilibrium for this graph-constrained model and the core of the respective NTU coalitional game is a new finding. We also proved the more strict \emph{strong stability} property. A special aspect of our model is that we do not consider side payments (money), not even in the form of budget constraints. This renders the analysis significantly different than most of the previous models \cite{osborne}, \cite{collaborativeconsumption}, \cite{zhang-proportional}, yet very appropriate for the considered problem. We believe that these results open many fascinating directions for future work. Among them, it is important to relax the common assumption of large demand that exceeds resource availability for the users (considered also in \cite{RJohariToNBilateral2011}, \cite{zhang-proportional}, \cite{saad-smart}, \cite{KearnsGraphEcon2004}, \cite{KearnsEconSocial2004}, \cite{myerson-gametheory-book}, \cite{efstathiou}), and provide a formal proof for the convergence of the dynamic asynchronous user interaction model.


\section*{Acknowledgements}
\vspace{-1mm}
This research has been co-financed by the European Union (European Social Fund - ESF) and Greek national funds through the Operational Program ''Education and Lifelong Learning'' of the National Strategic Reference Framework (NSRF) - Research Funding Program: Thales, Investing in knowledge society through the European Social Fund, for the project SOCONET. 

%
%

\section*{Appendix} 

We provide the additional proofs for the theorems and the lemmas. 

\subsection*{Proof of Theorem \ref{lem:reccip1}}

We begin with \textbf{(i)}. Let $\bar{\boldsymbol{d}}$ be a lex-optimal allocation. Starting from $\bar{\boldsymbol{d}}$ we will construct $\boldsymbol{d}$ by a sequence of link resource reallocations. Notice that for any allocation $\boldsymbol{\hat{d}}$, it holds:
\begin{equation}
\hat{r}_{i} =\hat{l}_{\hat{k}(i)}D_{i}=\sum_{j\in{\cal N}_{i}}\left(\hat{l}_{\hat{k}(i)}\hat{d}_{ij}\right)=\sum_{j\in{\cal N}_{i}}\hat{d}_{ji}
\end{equation}
where the first equality is by definition of $\hat{l}_{\hat{k}(i)}$, the second by (\ref{eq:Allloc}), and the last by definition of $r_{i}$. Hence under any allocation $\hat{\boldsymbol{d}}$ we have the following fact.
\begin{align}
&\hat{l}_{\hat{k}(i)}\hat{d}_{ij_{1}}>\hat{d}_{j_{1}i}\mbox{ for some }j_{1}\in{\cal N}_{i},\nonumber \\
&\mbox{ iff }\hat{l}_{k(i)}\hat{d}_{ij_{2}}<\hat{d}_{j_{2}i}\mbox{ for another }j_{2}\in{\cal N}_{i}.\label{eq:iff}
\end{align}
Let ${\cal Y}_{0}$ be the set of links $\left(i,j\right)$ for which it holds $\,\bar{l}_{\bar{k}(i)}\bar{d}_{ij}>\bar{d}_{ji}$. If ${\cal Y}_{0}$
is nonempty, we will show that we can reallocate resource so that under the new allocaton $\boldsymbol{d}_{1}$, the resources the nodes receive remain the same, while for at least one of the links $(i,j)$ in ${\cal Y}_{0}$ it holds, $\bar{l}_{\bar{k}(i)}\bar{d}_{ij}=\bar{d}_{ji}$ while the rest of the inequalities still hold in their original direction. Hence the new allocation $\boldsymbol{d}_{1}$ is also lex-optimal, while ${\cal Y}_{1}={\cal Y}_{0}-\{\left(i,j\right)\}$. Proceeding in this manner we will arrive at an allocation $\boldsymbol{d}$ which is still lex-optimal but for which ${\cal Y}=\emptyset$. Based on (\ref{eq:iff}) we will then conclude that the last allocation satisfies
\begin{equation}
d_{ji}^{*}/d_{ij}^{*}=r_{i}^{*}/D_i=l_{k(i)}^{*},\, j\in \mathcal{N}_{i},
\end{equation}
as stated in \textbf{(i)}. Also, let $(i_{1},i_{0})$ be a link such that $\bar{l}_{\bar{k}(i_{1})}\bar{d}_{i_{1}i_{0}}>\bar{d}_{i_{0}i_{1}}\geq0$ (hence $\bar{l}_{\bar{k}(i_{1})}>0$). Then according to (\ref{eq:iff}) there must be a link $(i_{1},i_{2})$ such that $\bar{l}_{\bar{k}(i_{1})}\bar{d}_{i_{1}i_{2}}<\bar{d}_{i_{2}i_{1}}$ or
\[
\frac{1}{\bar{l}_{\bar{k}(i_{1})}}\bar{d}_{i_{2}i_{1}}>\bar{d}_{i_{1}i_{2}}\geq0.
\]
But due to lex optimality of $\boldsymbol{\bar{d}}$, we conclude from Corollary \ref{thm:MainLex} \textcolor{black}{Property \ref{enu:CorItem5}} that $\bar{l}_{\bar{k}(i_{2})}=\left(1/\bar{l}_{\bar{k}(i_{1})}\right)$ and hence the above becomes $\bar{l}_{\bar{k}(i_{2})}\bar{d}_{i_{2}i_{1}}>\bar{d}_{i_{1}i_{2}}$. Repeating this procedure we find a sequence of links $\left(i_{m},i_{m-1}\right),\ m=1,2,..$ for which it holds:
\begin{equation}
\bar{l}_{\bar{k}(i_{m})}\bar{d}_{i_{m}i_{m-1}}>\bar{d}_{i_{m-1}i_{m}}.\label{eq:cycle}
\end{equation}
Since the number of nodes is finite, we will eventually find a simple
(no repeated nodes) cycle that satisfies (\ref{eq:cycle}). For all
nodes $m=1,2,..M$ on this cycle, subtract resource $\delta$ from $\bar{d}_{i_{m}i_{m-1}}$
and increase by $\delta$ the resource $\bar{d}_{i_{m}i_{m+1}}.$ In addition,
we require that the following relation must be satisfied for all nodes
$m$ on the cycle,
\begin{equation}
\bar{l}_{i_{m}}(\bar{d}_{i_{m}i_{m-1}}-\delta)\geq\bar{d}_{i_{m-1}i_{m}}+\delta,\ m=1,2,...,M,\,\text{or}
\end{equation}
\[
0<\delta\leq\frac{\bar{l}_{i_{m}}\bar{d}_{i_{m}i_{m-1}}-\bar{d}_{i_{m-1}i_{m}}}{\bar{l}_{i_{m}}+1},\ m=1,2,...,M
\]
This choice of $\delta$ ensures that the increase-decrease of resource
allocation gives a new allocation and that with the new allocation
inequalities (\ref{eq:cycle}) either still hold in their original
direction or become equalities Since we have a cycle, this increase-decrease
does not alter the resource the nodes of the cycle get. Since the resources
of the rest of the nodes in the network are not changed, the resulting
allocation is still lex-optimal. We now pick
\[
\delta_{1}=\max_{m}\left\{ \frac{\bar{l}_{i_{m}}\bar{d}_{i_{m}i_{m-1}}-\bar{d}_{i_{m-1}i_{m}}}{\bar{l}_{i_{m}}+1}\right\} >0.
\]
This choice ensures that at least one of the inequalities (\ref{eq:cycle})
become equality for some node in the cycle as desired.

Next we prove \textbf{(ii)}. First, we need a useful result: if an allocation $\bm{d}$ satisfies $d_{ji}^{*}=d_{ij}^{*}l_{k(i)}^{*},\, j\in \mathcal{N}_{i}$, then for any node $j\in\mathcal{D}_i$ it holds:
\begin{equation}
l_{k(j)}>0,\,\,\text{and}\,\,l_{k(j)}=1/l_{k(i)} \label{eq:itemreceip22}
\end{equation}
To see this, note that since $j\in{\cal D}_{i},$ by definition $d_{ij}>0$, hence $l_{k(j)}=\left(r_{j}/D_{j}\right)\geq\left(d_{ij}/D_{j}\right)>0$ and since by assumption
\begin{equation}
d_{ij}=l_{k(j)}d_{ji},\label{eq:help1}
\end{equation}
it also holds $d_{ji}>0$. Next, since by assumption
\begin{equation}
d_{ji}=l_{k(i)}d_{ij},\label{eq:help2}
\end{equation}
multiplying (\ref{eq:help1}) (\ref{eq:help2}) and canceling (the nonzero) terms we have $l_{k_{i}}l_{k_{j}}=1.$ Therefore (\ref{eq:help1}) holds.

Now, notice first that any node $i$ gives resource to at least one node $j$, hence $d_{ij}>0$. Therefore, it follows from eq. (\ref{eq:itemreceip22}) that $l_{1}>0$. If under allocation $\boldsymbol{d}$ there is only one level, i.e., $K=1$, then by Proposition \ref{prop:LexK1}, it is lex-optimal. Hence we concentrate on the case $K\ge2$. We will show that allocation $\boldsymbol{d}$ satisfies the properties of Theorem \ref{thm:MainTh0} and hence, by Theorem \ref{thm:MainLexSuff}, it is lex-optimal.

Consider first $k=1$. The nodes in ${\cal L}_{1}$ constitute an independent set. To see this note that $d_{ij}=0$ for all links $(i,j)$ with $i,j\in{\cal L}_{1}$ since otherwise (i.e., $d_{ij}>0$) by eq. (\ref{eq:itemreceip22}), and the fact that $l_{k(1)}=l_{k(2)}=l_{1}$ it will follow that $l_{1}=1$, which contradicts Lemma \ref{lem:levels}.
Hence all nodes in ${\cal L}_{1}$ give their resource to nodes in higher layers. Using again eq. (\ref{eq:itemreceip22}), we conclude that all nodes in ${\cal L}_{1}$ give resource to nodes at level with value $1/l_{1}$. Now, if there is a link $(i,j)\in{\cal E}$ with $i,j\in{\cal L}_{1}$ then since $d_{ij}=0$, the condition in \textbf{(i)}, i.e., "the neighbors not receiving resource from $i$ have higher exchange ratio" (which holds according to Lemma \ref{lem:neighbor}), implies that $l_{k(j)}\geq1/l_{1}$ and since $l_{k(j)}=l_{1},$ we conclude $l_{1}\geq1$, a contradiction. Hence Item \ref{enu:MainTh0Item1} of Theorem \ref{thm:MainTh0} holds for $k=1.$

Consider now the nodes in ${\cal N}\left({\cal L}_{1}\right)$. Nodes in ${\cal N}\left({\cal L}_{1}\right)$ give resource only to nodes in ${\cal L}_{1}$. To see this, note that if node $j\in{\cal N}\left({\cal L}_{1}\right)$ were giving resource to a node $i\notin{\cal L}_{1}$, then, since there are neighbors of $j$ in ${\cal L}_{1}$ by a similar reasoning ("the neighbors not receiving resource from $i$ have higher exchange ratio", Lemma \ref{lem:neighbor}), it would hold $l_{1}\geq l_{k(i)}$, i.e., $l_{1}=l_{k(i)}$, which contradicts the fact that $i\notin{\cal L}_{1}$. Since all nodes in ${\cal N}\left({\cal L}_{1}\right)$ give resource to nodes in ${\cal L}_{1}$ it follows from (\ref{eq:itemreceip22}) that all nodes in ${\cal N}\left({\cal L}_{1}\right)$ are at the same level and $l_{k\left({\cal N}\left({\cal L}_{1}\right)\right)}=1/l_{1}$.

We claim now that for any node $i\in{\cal N}-{\cal N}\left({\cal L}_{1}\right)$ it holds, $l_{k(i)}<l_{k\left({\cal N}\left({\cal L}_{1}\right)\right)}$
which implies that ${\cal L}_{K}={\cal N}\left({\cal L}_{1}\right)$ and hence Item \ref{enu:MainTh0Item2} of Theorem \ref{thm:MainTh0} holds for $k=1.$ Indeed, assume that $l_{k(i)}\geq l_{k\left({\cal N}\left({\cal L}_{1}\right)\right)}=1/l_{1}$. Then since node $i$ gives resource to at least another neighbor node $j$, by (\ref{eq:itemreceip22}) we would have $l_{k(j)}=\left(1/l_{k(i)}\right)\leq l_{1}$ hence $l_{k(j)}=l_{1}$, i.e., $j\in{\cal L}_{1}$ which contradicts the fact that $i\in{\cal N}-{\cal N}\left({\cal L}_{1}\right)$.

The fact that Item \ref{enu:MainTh0Item2.1} of Theorem \ref{thm:MainTh0} holds for $k=1,$ follows again from (\ref{eq:itemreceip22}). Also, Item \ref{enu:MainTh0Item3} of Theorem \ref{thm:MainTh0} holds since as shown above, all nodes in ${\cal L}_{1}=$${\cal N}\left({\cal L}_{1}\right)$ give their resource to nodes in ${\cal L}_{1}$. If $K=2$, the lex- optimality of $\boldsymbol{d}$ follows from Theorem \ref{thm:MainLexSuff}. Consider next $K\geq3$. According to Lemma \ref{lem:InOutZero}, it holds
\[
{\rm In}\left({\cal L}_{k}\cup{\cal L}_{K-k+1}\right)={\rm Out}\left({\cal L}_{k}\cup{\cal L}_{K-k+1}\right)=0,
\]
 hence the restriction of $\boldsymbol{d}$ on ${\cal Q}_{2}$, $\boldsymbol{d}_{{\cal Q}_{2}}$ constitutes an allocation on $G_{{\cal Q}_{2}}$with $K_{{\cal Q}_{2}}=K-2$ levels. Moreover, since no nodes at levels ${\cal L}_{1}$ and ${\cal L}_{K}$ are in ${\cal Q}_{2},$ we have ${\cal L}_{{\cal Q}_{2},1}={\cal L}_{2}$, and ${\cal L}_{{\cal Q}_{2},K_{{\cal Q}_{2}}}={\cal L}_{K-1}$. Also, (\ref{eq:itemreceip22}) and \textbf{(ii)} continue to hold for $\boldsymbol{d}_{{\cal Q}_{2}}$ on $G_{{\cal Q}_{2}}$.
If $K=3,$ then we have $K_{{\cal Q}_{2}}=1,$ hence Item of Theorem \ref{thm:MainTh0} holds. If $K\geq4,$ we can apply now the arguments we used for $k=1$ to show that Items \ref{enu:MainTh0Item1}-\ref{enu:MainTh0Item3} hold for $k=2.$ Proceeding iteratively we show that all properties in Theorem \ref{thm:MainTh0} hold for $\boldsymbol{d}$ and hence it is lex-optimal. $\blacksquare$

\subsection*{PROOF of Theorem \ref{thm:MainLexSuff}}
Before proving Theorem \ref{thm:MainLexSuff}, we need the following lemma.

\begin{lemma}
\label{lem:InOutZero}
If an allocation $\boldsymbol{d}$ with $K\geq2$ satisfies Properties \ref{enu:MainTh0Item1}- \ref{enu:MainTh0Item3} of Theorem \ref{thm:MainTh0}, it holds for $k=1,...,\flr{K/2}$:
\begin{equation}
{\rm In}\left({\cal L}_{k}\cup{\cal L}_{K-k+1}\right)={\rm Out}\left({\cal L}_{k}\cup{\cal L}_{K-k+1}\right)=0.\label{eq:InOutZero}
\end{equation}
\end{lemma}
\begin{proof}
Let $k=1$. Since by Property \ref{enu:MainTh0Item1} of Theorem \ref{thm:MainTh0} the set ${\cal L}_{1}$ is independent, we have $\sum_{i\in{\cal L}_{1}}r_{i}={\rm In}\left({\cal L}_{1}\right)$.
Also, since only nodes in ${\cal N}_{Q_{1}}\left({\cal L}_{1}\right)={\cal N}\left({\cal L}_{1}\right)={\cal L}_{K}$ may have links with nodes in $\left({\cal L}_{1}\cup{\cal L}_{K}\right)^{c}$, we have ${\rm Out}\left({\cal L}_{K}\right)={\rm Out}\left({\cal L}_{1}\cup{\cal L}_{K}\right)+{\rm In}\left({\cal L}_{1}\right)$, and hence:
\begin{align*}
&\sum_{i\in{\cal L}_{K}}D_{i}  \geq{\rm Out}\left({\cal L}_{K}\right)={\rm Out}\left({\cal L}_{1}\cup{\cal L}_{K}\right)+{\rm In}\left({\cal L}_{1}\right)\\
 & ={\rm Out}\left({\cal L}_{1}\cup{\cal L}_{K}\right)+\sum_{i\in{\cal L}_{1}}r_{i}={\rm Out}\left({\cal L}_{1}\cup{\cal L}_{K}\right)+\sum_{i\in{\cal L}_{K}}D_{i},
\end{align*}
where the last equality is due to Property \ref{enu:MainTh0Item3} of Th. \ref{thm:MainTh0}. The last equality implies that ${\rm Out}\left({\cal L}_{1}\cup{\cal L}_{K}\right)=0$ and:
\begin{equation}
{\rm Out}\left({\cal L}_{K}\right)=\sum_{i\in{\cal L}_{K}}D_{i}.\label{eq:InOut1}
\end{equation}
Next, we have for the nodes in ${\cal L}_{K}$:
\[
\sum_{i\in{\cal L}_{K}}r_{i}+{\rm Out}\left({\cal L}_{K}\right)=\sum_{{\cal L}_{K}}D_{i}+{\rm In}\left({\cal L}_{K}\right)
\]
Since by independence of ${\cal L}_{1}$ it holds:
\[
{\rm In}\left({\cal L}_{K}\right)=\sum_{i\in{\cal L}_{1}}D_{i}+{\rm In}\left({\cal L}_{1}\cup{\cal L}_{K}\right),
\]
and taking into account (\ref{eq:InOut1}), we conclude from (\ref{eq:equality}) $\sum_{i\in{\cal L}_{K}}r_{i}=\sum_{i\in{\cal L}_{1}}D_{i}+{\rm In}\left({\cal L}_{k}\cup{\cal L}_{K}\right)$,
or since $r_{i}=l_{K}D_{i},\ i\in{\cal L}_{K}$, it is:
\[
l_{K}\sum_{i\in{\cal L}_{K}}D_{i}=\sum_{i\in{\cal L}_{1}}D_{i}+{\rm In}\left({\cal L}_{k}\cup{\cal N}_{Q_{k}}\left({\cal L}_{k}\right)\right).
\]
Similarly, from the equality in Property \ref{enu:MainTh0Item3} we have $l_{1}\sum_{i\in{\cal L}_{1}}D_{i}=\sum_{i\in{\cal L}_{K}}D_{i}$. Multiplying the last two equalities and rearranging terms we get:
\begin{equation}
l_{1}l_{K}=1+\frac{{\rm In}\left({\cal L}_{1}\cup{\cal L}_{K}\right)}{\left(\sum_{i\in{\cal L}_{1}}D_{i}\right)\left(\sum_{i\in{\cal L}_{K}}r_{i}\right)}\,.
\end{equation}
But since by Item \ref{enu:MainTh0Item2.1} of Theorem \ref{thm:MainTh0} it hold $l_{1}l_{K}=1,$ we conclude that ${\rm In}\left({\cal L}_{1}\cup{\cal L}_{K}\right)=0$.

Hence, if $K\in\left\{ 2,3\right\} $ the lemma holds. Next, assume $K\geq4$ and observe that since (\ref{eq:InOutZero}) holds for $k=1$, the restriction $\boldsymbol{d}_{{\cal Q}_{2}}$ of $\boldsymbol{d}$ to ${\cal Q}_{2}$ is an allocation on ${\cal Q}_{2}$ with $K-2$ levels and by construction ${\cal L}_{{\cal Q}_{2},1}={\cal L}_{2}$, ${\cal L}_{{\cal Q}_{2},K-2}={\cal L}_{K-1}={\cal L}_{K-2+1}$. Therefore, we can repeat the arguments above for $k=2$ and inductively show that the lemma holds for $k=1,....,\flr{K^{*}/2}$. $\blacksquare$
\end{proof}

\uline{Proof of Theorem \ref{thm:MainLexSuff}}: For $k=1,$ since by Properties \ref{enu:MainTh0Item1} and \ref{enu:MainTh0Item3} of Theorem \ref{thm:MainTh0} $\bar{{\cal L}_{1}}$ is an independent set and $\sum_{i\in{\cal L}_{1}}r_{1}^{*}=\sum_{i\in{\cal N}\left({\cal L}_{1}\right)}D_{i}$,
it follows from Lemma \ref{lem:MaxMin} that ${\cal L}_{1}^{*}={\cal L}_{1}$ and $l_{1}^{*}=l_{1}$. Also, by Properties \ref{enu:MainTh0Item2}, \ref{enu:MainTh0Item2.1} of Theorem \ref{thm:MainTh0} we have $l_{K}=l_{K^{*}}^{*}$ and ${\cal L}_{K}={\cal N}\left({\cal L}_{1}\right)={\cal N}\left({\cal L}_{1}^{*}\right)={\cal L}_{K^{*}}^{*}$, where the last equality follows from Proposition \ref{lem:MainLem}. If $K=2,$ then since ${\cal N}={\cal L}_{1}\cup{\cal L}_{2}={\cal L}_{1}^{*}\cup{\cal L}_{K^{*}}^{*}$ we have necessarily $K^{*}=2$ and we conclude that $\boldsymbol{d}$ is lex-optimal. Assume now that $K=3$. From Proposition \ref{lem:MainLem} we then have ${\rm In}\left({\cal L}_{1}\cup{\cal L}_{3}\right)={\rm Out}\left({\cal L}_{2}\cup{\cal L}_{3}\right)=0$.
Hence the restriction $\boldsymbol{d}_{{\cal Q}_{2}}$ is an allocation on $G_{{\cal Q}_{2}}$ with $K_{{\cal Q}_{2}}=1.$ It follows by Proposition \ref{prop:LexK1} that $\boldsymbol{d}_{{\cal Q}_{2}}$ is lex-optimal in $G_{{\cal Q}_{2}}$ and $l_{2}=1$. This implies that any lex-optimal allocation allocation on $G_{{\cal Q}_{2}}$ has $K_{{\cal Q}_{2}}^{*}=1$ and $l_{2}^{*}=1$. We then conclude that $K^{*}=3$ and arguing as in the case $K=2$, that $\boldsymbol{d}$ is lex-optimal.

We will use induction to show the Theorem for allocations $\boldsymbol{d}$ with arbitray $K$. Assume that the theorem holds for allocation with up to $K-1,$ $K\geq4$ levels and let next $K\geq4$. By (\ref{eq:MainLemInOut0}), the vector $\boldsymbol{d}_{{\cal Q}_{2}}$ constitutes an allocation on graph $G_{{\cal Q}_{2}}.$ Since this allocation has $K-2$ levels
we can apply the inductive hypothesis to conclude that the allocation $\boldsymbol{d}_{{\cal Q}_{2}}$ is lex-optimal in $G_{{\cal Q}_{2}}$. But the same holds for the restriction $\hat{\boldsymbol{d}}_{{\cal Q}_{2}}$ to $G_{{\cal Q}_{2}}$, of any lex-optimal allocation $\hat{\boldsymbol{d}}$. By uniqueness of lex optimality we conclude that all levels $l_{k},$ level sets ${\cal L}_{k}$ and received resource $r_{i}$ of $\boldsymbol{d}_{{\cal Q}_{2}}$ for $k=2,...,K-1$ are identical to those of any lex-optimal allocation. It follows that $K=K^{*}$ and we already showed that ${\cal L}_{1}^{*}={\cal L}_{1}$, $l_{1}=l_{1}^{*}$ ${\cal L}_{K}={\cal L}_{K^{*}}^{*},$ and $l_{K}=l_{K^{*}}^{*}.$ The lex-optimality of $\boldsymbol{d}$ follows. $\blacksquare$

\section*{Proof of Theorem 4}

Corollary \ref{thm:MainLex} is a simple consequence of the properties of the lex-optimal policies and can be derived by combining Theorems \ref{thm:MainTh0} and \ref{thm:Stability}. Hence, we only need to focus on the main result of this subsection, i.e., Theorem \ref{thm:Stability}.
\begin{proof}
Based on the results of Theorem \ref{thm:MainTh0} (and using the notation of Corollary \ref{thm:MainLex}), let ${\cal M}_{k}^{*}={\cal L}_{k}^{*}\cup{\cal L}_{K-k+1}^{*},\ k=1,2,...,\flr{K^{*}/2}$ be the formed groups under $\boldsymbol{d}$; if $K^{*}$ is odd, there is also a group ${\cal M}_{\ceil{K^{*}/2}}^{*}={\cal L}_{\ceil{K^{*}/2}}^{*}$. Below it will help to denote ${\cal L}_{1,k}\triangleq{\cal L}_{k}^{*}$
and ${\cal L}_{2,k}\triangleq{\cal L}_{N-k+1}^{*},$ $1\leq k\leq\flr{K^{*}/2}$ so that ${\cal M}_{k}^{*}={\cal L}_{1,k}\cup{\cal L}_{2,k}$, $1\leq k\leq\flr{K^{*}/2}$. We also define $l_{1,k}\triangleq l_{k}^{*}$ and $l_{2,k}\triangleq l_{N-k+1}^{*},\ 1\leq k\leq\flr{K^{*}/2}.$ If $K^{*}$ is odd then define ${\cal L}_{1,\ceil{K^{*}/2}}=\emptyset$ ${\cal L}_{2,\ceil{K^{*}/2}}={\cal L}_{\ceil{K^{*}/2}}^{*}$, $l_{1,\ceil{K^{*}/2}}=l_{2,\ceil{K^{*}/2}}=l_{\ceil{K^{*}/2}}^{*}$=1.

Consider an arbitrary \emph{nonempty} set of nodes ${\cal C}$ and define ${\cal C}_{1,k}={\cal C}\cap{\cal L}_{1,k}$, ${\cal C}_{2,k}={\cal C}\cap{\cal L}_{2,k},$ $1\leq k\leq\flr{K^{*}/2}$, and in case $K^{*}$ is odd, ${\cal C}_{2,\ceil{K^{*}/2}}={\cal C}\cap{\cal L}_{\ceil{K^{*}/2}}$. Hence ${\cal C}_{1,k}\cup{\cal C}_{2,k}={\cal M}_{k}^{*}\cap{\cal C}.$ Let $\hat{\boldsymbol{d}}$ be an allocation on this set such that $\hat{r}_{i}\geq r_{i}^{*}$ for all $i\in{\cal C}$ and $\hat{r}_{j_{0}}>r_{j_{0}}^{*}$ for some $j_{0}\in{\cal C}$. Below we argue by contradiction that such set does not exist. In the case where the induced subgraph contains singletons, the results is trivial.

From Theorem \ref{thm:MainTh0} and specifically the properties of the lex-optimal allocations, we know that the nodes in the set $\mathcal{L}_{k}^{*}$, for $1\leq k\leq K^*/2$ may be connected only to nodes in sets $\mathcal{L}_{K-m+1}^{*}$, with $1\leq m\leq k$. Hence, we have the following properties
\begin{enumerate}
\item \label{enu:Prop1}Nodes in ${\cal L}_{1,k}\ 1\leq k\leq\flr{K^{*}/2}$ may be connected to nodes in ${\cal L}_{2,m},\ 1\leq m\leq k.$
\item \label{enu:Prop2}Nodes in ${\cal L}_{2,k},\ 1\leq k\leq\flr{K^{*}/2}$ may be connected to nodes in all sets ${\cal L}_{2,m},\ 1\leq m\leq\ceil{K^{*}/2}$
and to nodes in the sets ${\cal L}_{1,m},\ k\le m\leq\flr{K^{*}/2}.$
\item \label{enu:Prop3}If $K^{*}$ is odd, then nodes in ${\cal L}_{2,\ceil{K^{*}/2}}$ may be connected to nodes in all sets ${\cal L}_{2,m},\ 1\leq m\leq\ceil{K^{*}/2}$.
\end{enumerate}

Under allocation $\hat{\boldsymbol{d}}$, let $a_{(t,k)}^{(h,m)}$ be the proportion of offered resource by the nodes in ${\cal C}_{t,k}$ ( i.e., $\sum_{i\in{\cal C}_{t,k}}D_{i}$ ) to nodes in ${\cal C}_{h,m}$. From Properties \ref{enu:Prop1}) and \ref{enu:Prop2}) above we then have for any $k,$ $1\leq k\leq\flr{K^{*}/2}$
\begin{eqnarray}
\sum_{m=1}^{k}a_{(1,k)}^{(2,m)} & = & 1\label{eq:firstSum}\\
\sum_{m=1}^{\ceil{K^{*}/2}}a_{(2,k)}^{(2,m)}+\sum_{m=k}^{\flr{K^{*}/2}}a_{(2,k)}^{(1,m)} & = & 1,\label{eq:secondSum}
\end{eqnarray}
and from Property \ref{enu:Prop3}, if $K^{*}$ is odd,
\begin{equation}
\sum_{m=1}^{\ceil{K^{*}/2}}a_{(2,\ceil{K^{*}/2})}^{(2,m)}=1.\label{eq:ThirdSum}
\end{equation}

Since nodes in ${\cal C}_{1,k},\ 1\leq k\leq\flr{K^{*}/2}$ may be connected and hence get their resource from nodes in ${\cal L}_{2,m},\ 1\leq m\leq k,$ we have, for every $1\leq k\leq\flr{K^{*}/2}$:
\begin{eqnarray}
\sum_{m=1}^{k}a_{(2,m)}^{(1,k)}\sum_{i\in{\cal C}_{2,m}}D_{i} = \sum_{i\in{\cal C}_{1,k}}\hat{r}_{i}\geq l_{1,k}\sum_{i\in{\cal C}_{1,k}}D_{i},\label{eq:FirstIneq}
\end{eqnarray}
with strict inequality holding if $j_{0}\in{\cal C}_{1,k}$ for some $k,\ 1\leq k\leq\flr{K^{*}/2}.$

Similarly, since nodes in ${\cal C}_{2,k},\ 1\leq k\leq\flr{K^{*}/2}$ may be connected and hence get their resource from nodes in ${\cal L}_{2,m},\ 1\leq m\leq\ceil{K^{*}/2}$
and from nodes in the sets ${\cal L}_{1,m},\ k\le m\leq\flr{K^{*}/2}$, it holds:
\begin{eqnarray}
\sum_{m=1}^{\ceil{K^{*}/2}}a_{(2,m)}^{(2,k)}\sum_{i\in{\cal C}_{2,m}}D_{i}+\sum_{m=k}^{\flr{K^{*}/2}}a_{(1,m)}^{(2,k)}\sum_{i\in{\cal C}_{1,m}}D_{i} \nonumber \\
=\sum_{i\in{\cal C}_{2,k}}\hat{r}_{i} \geq l_{2,k}\sum_{i\in{\cal C}_{2,k}}D_{i},\ 1\leq k\leq\flr{K^{*}/2},\label{eq:SecondIneq-1}
\end{eqnarray}
with strict inequality holding if $j_{0}\in{\cal C}_{2,k}$ for some $k,\ 1\leq k\leq\flr{K^{*}/2}.$

For a given $k$, multiplying (\ref{eq:SecondIneq-1}) by $l_{1,k}$, adding (\ref{eq:FirstIneq}), and considering that $l_{1,k}l_{2,k}=l_{k}^{*}l_{N-k+1}^{*}=1$, we get:
\begin{align}
&\sum_{m=1}^{k}a_{(2,m)}^{(1,k)}\sum_{i\in{\cal C}_{2,m}}D_{i}+l_{1,k}\sum_{m=1}^{\ceil{K^{*}/2}}a_{(2,m)}^{(2,k)}\sum_{i\in{\cal C}_{2,m}}D_{i}\nonumber \\ &+l_{1,k}\sum_{m=k}^{\flr{K^{*}/2}}a_{(1,m)}^{(2,k)}\sum_{i\in{\cal C}_{1,m}}D_{i}\nonumber \\
&\geq l_{1,k}\sum_{i\in{\cal C}_{1,k}}D_{i}+\sum_{i\in{\cal C}_{2,k}}D_{i},\ k=1,...,\flr{K^{*}/2},\label{eq:BasicIneq}
\end{align}
with strict inequality holding if $j_{0}\in{\cal C}_{1,k}\cup{\cal C}_{2,k}$ for some $k,\ 1\leq k\leq\flr{K^{*}/2}.$ Adding the inequalities in (\ref{eq:BasicIneq}) we get
\begin{align}
&\sum_{k=1}^{\flr{K^{*}/2}}\sum_{m=1,}^{k}a_{(2,m)}^{(1,k)}\sum_{i\in{\cal C}_{2,m}}D_{i}\nonumber \\ &+\sum_{k=1}^{\flr{K^{*}/2}}l_{1,k}\sum_{m=1}^{\ceil{K^{*}/2}}a_{(2,m)}^{(2,k)}\sum_{i\in{\cal C}_{2,m}}D_{i} \nonumber \\
&+\sum_{k=1}^{\flr{K^{*}/2}}l_{1,k}\sum_{m=k}^{\flr{K^{*}/2}}a_{(1,m)}^{(2,k)}\sum_{i\in{\cal C}_{1,m}}D_{i}\nonumber \\
&\geq\sum_{k=1}^{\flr{K^{*}/2}}l_{1,k}\sum_{i\in{\cal C}_{1,k}}D_{i}+\sum_{k=1}^{\flr{K^{*}/2}}\sum_{i\in{\cal C}_{2,k}}D_{i},\label{eq:BasicIneq-1}
\end{align}
with strict inequality holding if $j_{0}\in\cup_{k=1}^{\flr{K^{*}/2}}\left({\cal C}_{1,k}\cup{\cal C}_{2,k}\right).$
Note now that:
\begin{align}
&\sum_{k=1}^{\flr{\frac{K^{*}}{2}}}\sum_{m=1}^{k}a_{(2,m)}^{(1,k)}\sum_{i\in{\cal C}_{2,m}}D_{i}=\nonumber \\
&=\sum_{m=1}^{\flr{\frac{K^{*}}{2}}}\big(\sum_{k=m}^{\flr{\frac{K^{*}}{2}}}a_{(2,m)}^{(1,k)}\big)\sum_{i\in{\cal C}_{2,m}}D_{i}\label{eq:interch1}
\end{align}
\begin{align}
&\sum_{k=1}^{\flr{\frac{K^{*}}{2}}}l_{1,k}\sum_{m=k}^{\flr{\frac{K^{*}}{2}}}a_{(1,m)}^{(2,k)}\sum_{i\in{\cal C}_{1,m}}D_{i}=\nonumber \\
&=\sum_{m=1}^{\flr{\frac{K^{*}}{2}}}\big(\sum_{k=1}^{m}l_{1,k}a_{(1,m)}^{(2,k)}\big)\sum_{i\in{\cal C}_{1,m}}D_{i}\label{eq:interch3}
\end{align}
where we have applied the identity
\begin{equation}
\sum_{k=1}^{K}\sum_{m=1}^{k}a_{km}=\sum_{m=1}^{K}\sum_{k=m}^{K}a_{km}
\end{equation}
Also, if $K^{*}$ is even, then since $\flr{K^{*}/2}=\ceil{K^{*}/2}$,
\begin{align}
&\sum_{k=1}^{\flr{\frac{K^{*}}{2}}}l_{1,k}\sum_{m=1}^{\ceil{K^{*}/2}}a_{(2,m)}^{(2,k)}\sum_{i\in{\cal C}_{2,m}}D_{i}= \nonumber \\
&=\sum_{m=1}^{\flr{\frac{K^{*}}{2}}}\big(\sum_{k=1}^{\flr{\frac{K^{*}}{2}}}l_{1,k}a_{(2,m)}^{(2,k)}\big)\sum_{i\in{\cal C}_{2,m}}D_{i},\label{eq:interch2}
\end{align}
while if $K^{*}$is odd,
\begin{align}
&\sum_{k=1}^{\flr{\frac{K^{*}}{2}}}l_{1,k}\sum_{m=1}^{\ceil{ \frac{K^{*}}{2} }}a_{(2,m)}^{(2,k)}\sum_{i\in{\cal C}_{2,m}}D_{i}=\sum_{m=1}^{\flr{\frac{K^{*}}{2}}}\big(\sum_{k=1}^{\flr{\frac{K^{*}}{2}}}l_{1,k}a_{(2,m)}^{(2,k)}\big)\cdot \nonumber\\ &\cdot \sum_{i\in{\cal C}_{2,m}}D_{i}+\sum_{k=1}^{\flr{\frac{K^{*}}{2}}}l_{1,k}a_{(2,\ceil{ \frac{K^{*}}{2} })}^{(2,k)}\sum_{i\in{\cal C}_{2,\ceil{\frac{K^{*}}{2}}}}D_{i}     \label{eq:interch2Odd}
\end{align}

where we applied the identity
\begin{equation}
\sum_{k=1}^{K}\sum_{m=1}^{K+1}a_{km}=\sum_{m=1}^{K}\sum_{k=m}^{K}a_{km}+\sum_{k=1}^{K}a_{k,(K+1)}.
\end{equation}

Assume now that $K^{*}$ is even. Using equalities (\ref{eq:interch1}), (\ref{eq:interch3}), and (\ref{eq:interch2}) in (\ref{eq:BasicIneq-1}), we get:
\begin{align}
&\sum_{m=1}^{\flr{K^{*}/2}}\left(\sum_{k=m}^{\flr{K^{*}/2}}a_{(2,m)}^{(1,k)}+\sum_{k=1}^{\flr{K^{*}/2}}l_{1,k}a_{(2,m)}^{(2,k)}\right)\sum_{i\in{\cal C}_{2,m}}D_{i}\nonumber \\
&+\sum_{m=1}^{\flr{K^{*}/2}}\left(\sum_{k=1}^{m}l_{1,k}a_{(1,m)}^{(2,k)}\right)\sum_{i\in{\cal C}_{1,m}}D_{i} \nonumber \\
&>\sum_{k=1}^{\flr{K^{*}/2}}l_{1,k}\sum_{i\in{\cal C}_{1,k}}D_{i}+\sum_{k=1}^{\flr{K^{*}/2}}\sum_{i\in{\cal C}_{2,k}}D_{i}\label{eq:FinalIneq}
\end{align}
where the inequality is strict since now
\begin{equation}
j_{0}\in{\cal C}=\cup_{k=1}^{\flr{K^{*}/2}}\left({\cal C}_{1,k}\cup{\cal C}_{2,k}\right). \nonumber
\end{equation}
But since $l_{1,k}<1$, $\forall\, k\in[1, \flr{\frac{K^{*}}{2}}]$, we have:
\begin{align}
&\sum_{k=m}^{\flr{K^{*}/2}}a_{(2,m)}^{(1,k)}+\sum_{k=1}^{\flr{K^{*}/2}}l_{1,k}a_{(2,m)}^{(2,k)}\leq   \\
&\leq \sum_{k=m}^{\flr{K^{*}/2}}a_{(2,m)}^{(1,k)}+\sum_{k=1}^{\flr{K^{*}/2}}a_{(2,m)}^{(2,k)}=1\ \ {\rm , by\ (\ref{eq:secondSum})} \nonumber
\end{align}
Taking into account that $l_{1,k}<l_{i,k'}$ if $k<k'$, we also have,


\begin{equation}
\sum_{k=1}^{m}l_{1,k}a_{(1,m)}^{(2,k)} \leq l_{1,m} \sum_{k=1}^{m}a_{(1,m)}^{(2,k)}= l_{1,m}\ \ {\rm by\ (\ref{eq:firstSum})}
\end{equation}
Hence, it holds:
\begin{align}
&\sum_{m=1}^{\flr{K^{*}/2}}\left(\sum_{k=m}^{\flr{K^{*}/2}}a_{(2,m)}^{(1,k)}+\sum_{k=1}^{\flr{K^{*}/2}}l_{1,k}\sum_{m=1}^{\flr{K^{*}/2}}a_{(2,m)}^{(2,k)}\right)\sum_{i\in{\cal C}_{2,m}}D_{i}\nonumber \\
&+\sum_{m=1}^{\flr{K^{*}/2}}\left(\sum_{k=1}^{m}l_{1,k}a_{(1,m)}^{(2,k)}\right)\sum_{i\in{\cal C}_{1,m}}D_{i} \nonumber \\
&\leq\sum_{m=1}^{\flr{K^{*}/2}}\sum_{i\in{\cal C}_{2,m}}D_{i}+\sum_{m=1}^{\flr{K^{*}/2}}l_{1,m}\sum_{i\in{\cal C}_{1,m}}D_{i}
\end{align}
which contradicts (\ref{eq:FinalIneq}).

It remains to consider the case that $K^{*}$ is odd. In this case, using equalities (\ref{eq:interch1}), (\ref{eq:interch3}), and (\ref{eq:interch2Odd}) in (\ref{eq:BasicIneq-1}), we obtain,
\begin{align}
&\sum_{m=1}^{\flr{K^{*}/2}}\left(\sum_{k=m}^{\flr{K^{*}/2}}a_{(2,m)}^{(1,k)}+\sum_{k=1}^{\flr{K^{*}/2}}l_{1,k}a_{(2,m)}^{(2,k)}\right)\sum_{i\in{\cal C}_{2,m}}D_{i} \nonumber \\
&+\sum_{m=1}^{\flr{K^{*}/2}}\left(\sum_{k=1}^{m}l_{1,k}a_{(1,m)}^{(2,k)}\right)\sum_{i\in{\cal C}_{1,m}}D_{i}\nonumber \\
&+\sum_{k=1}^{\flr{K^{*}/2}}l_{1,k}a_{(2,\ceil{K^{*}/2})}^{(2,k)}\sum_{i\in{\cal C}_{2,\ceil{K^{*}/2}}}D_{i}\nonumber \\
& \geq\sum_{k=1}^{\flr{K^{*}/2}}l_{1,k}\sum_{i\in{\cal C}_{1,k}}D_{i}+\sum_{k=1}^{\flr{K^{*}/2}}\sum_{i\in{\cal C}_{2,k}}D_{i}\label{eq:FinalIneq-1}
\end{align}
with strict inequality holding if $j_{0}\in\cup_{k=1}^{\flr{K^{*}/2}}\left({\cal C}_{1,k}\cup{\cal C}_{2,k}\right).$

Observe now that since nodes in ${\cal L}_{2,\ceil{K^{*}/2}}$ may be connected to nodes in all sets ${\cal L}_{2,m},\ 1\leq m\leq\ceil{K^{*}/2}$ we have,
\begin{align}
&\sum_{m=1}^{\ceil{K^{*}/2}}a_{(2,m)}^{(2,\ceil{K^{*}/2})}\sum_{i\in{\cal C}_{2,m}}D_{i} \nonumber \\
& =\sum_{i\in{\cal C}_{2,\ceil{K^{*}/2}}}\hat{r}_{i}\geq l_{1,\ceil{K^{*}/2}}\sum_{i\in{\cal C}_{1,\ceil{K^{*}/2}}}D_{i}\nonumber \\
& = \sum_{i\in{\cal C}_{1,\ceil{K^{*}/2}}}D_{i}\,,\mbox{ since \ensuremath{l_{1,\ceil{K^{*}/2}}=1.}}\label{eq:FinalIneq1}
\end{align}
with equality holding if $j_{0}\in{\cal C}_{1,\ceil{K^{*}/2}}.$ Adding (\ref{eq:FinalIneq-1}), (\ref{eq:FinalIneq1}):
\begin{align}
&\sum_{m=1}^{\flr{\frac{K^{*}}{2}}}\left(\sum_{k=m}^{\flr{\frac{K^{*}}{2}}}a_{(2,m)}^{(1,k)}+\sum_{k=1}^{\flr{\frac{K^{*}}{2}}}l_{1,k}a_{(2,m)}^{(2,k)}+a_{(2,m)}^{(1,\ceil{\frac{K^{*}}{2}})}\right)\sum_{i\in{\cal C}_{2,m}}D_{i} \nonumber \\
&+\sum_{m=1}^{\flr{\frac{K^{*}}{2}}}\left(\sum_{k=1}^{m}l_{1,k}a_{(1,m)}^{(2,k)}\right)\sum_{i\in{\cal C}_{1,m}}D_{i}+\sum_{k=1}^{\flr{\frac{K^{*}}{2}}}l_{1,k}a_{(2,\ceil{\frac{K^{*}}{2}})}^{(2,k)}\cdot  \nonumber \\
&\cdot \sum_{i\in{\cal C}_{2,\ceil{\frac{K^{*}}{2}}}}D_{i}>\sum_{k=1}^{\flr{\frac{K^{*}}{2}}}l_{1,k}\sum_{i\in{\cal C}_{1,k}}D_{i}+\sum_{k=1}^{\ceil{\frac{K^{*}}{2}}}\sum_{i\in{\cal C}_{2,k}}D_{i}
\end{align}
where the inequality is strict since $j_{0}\in{\cal C}=\cup_{k=1}^{\ceil{K^{*}/2}}\left({\cal C}_{1,k}\cup{\cal C}_{2,k}\right)$. Using again arguments similar to the case $K^{*}$ even, we arrive again at a contradiction. $\blacksquare$

\end{proof}


\vspace{-1mm}
\begin{thebibliography}{1}
\vspace{-1mm}
\bibitem{adalbdal} Adalbdal, www.adalbdal.com/\,.

\bibitem{strong-nash} N. Andelman, M. Feldman, Y. Mansour, ``Strong Price of Anarchy'', \emph{in Proc. of ACM SODA}, 2007.


\bibitem{anderson-soda} R. Anderson, I. Ashlagi, D. Gamarnik, and Y. Kanoria, ``A Dynamic Model of Barter Exachange'', \emph{in Proc. of ACM SODA}, 2015.

\bibitem{arrow-debreu} K. Arrow, and G. Debreu, ``Existence of an Equilibrium for a Competitive Economy'', \emph{Econom.}, 22(3), 1954.

\bibitem{RJohariToNBilateral2011} C. Aperjis, R. Johari, and M. Freedman, ``Bilateral and Multilateral Exchanges for Peer-Assisted Content Distribution'', \emph{IEEE/ACM Tran. on Netw.}, 19(5), 2011.

\bibitem{collaborativeconsumption} R. Botchman, and R. Roogers, ``The Rise of Collaborative Consumption'', \emph{HarperBusiness}, 2010.


\bibitem{hubaux-coop} L. Buttyan, J. Hubaux, ``Stimulating Cooperation in Self-organizing Mobile Ad Hoc Networks'', \emph{ACM J. Mobile Networks}, vol. 8, 2003.

\bibitem{confine}  B. Braem, et al., ``A Case for Research with and on Community Networks'', \emph{ACM CCR}, July 2013.

\bibitem{efstathiou} E. C. Efstathiou, et al., ``Controlled Wi-Fi Sharing in Cities: A Decentralized Approach Relying on Indirect Reciprocity'', \emph{IEEE Trans. on Mob. Comp.}, 9(8), 2010.

\bibitem{felson-coco-book} M. Felson, et al., ``Community Structure and Collaborative Consumption'', \emph{Amer. Behav. Sc.} 1978.

\bibitem{FON} FON, 2007, www.fon.com\,.

\bibitem{georgiadislexopt02} L. Georgiadis et al., ``Lexicographically Optimal Balanced Networks'', \emph{IEEE/ACM Transactions on Networking}, vol. 10, no. 5, 2002.

\bibitem{georgiadis-netgcoop} L. Georgiadis, G. Iosifidis, and L. Tassiulas, ``Dynamic Algorithms for Cooperation in User-provided Network Services'', \emph{in Proc. of NetGCoop}, 2014.

\bibitem{getrridapp} GetRidApp, http://getridapp.com/\,.

\bibitem{gridmates} GridMates, 2014, http://www.gridmates.com/\,.

\bibitem{Herings2000} P. Herings, G. V. Laan, D. Talman, ``Cooperative Games in Graph Structure'', \emph{Res. Memoranda, Maastricht}, vol. 3, no. 11, 2002.

\bibitem{homeexchange} \rev{HomeExchange, www.homeexchange.com}

\bibitem{JacksonWolinsky1996} M. Jackson, et al., ``A Strategic Model of Social and Economic Networks'', \emph{Journal of Econ. Th.}, 71(1), 1996.

\bibitem{KearnsGraphEcon2004} S. Kakade, et al., ``Graphical Economics'', \emph{in Springer Conf. on Learning}, 2004.

\bibitem{KearnsEconSocial2004} S. Kakade, et al., ``Economic Properties of Social Networks'', \emph{Advances in NIPS}, 2004.

\bibitem{ColellWhinstonGreenBook1995} A. Mas-Colell, et al., ``Microeconomic Theory'', \emph{Oxford Un. Press}, 1995.

\bibitem{ioannidis-peer-assisted} V. Misra, et al., ``Incentivizing Peer-assisted Services: A Shapley Value Approach'', \emph{ACM Sigmetrics}, 2010.

\bibitem{myerson-gametheory-book} R. Myerson, ``Game Theory: Analysis of Conflict'', \emph{Harvard Press}, 1997.

\bibitem{nace-tutorial} D. Nace, and M. Pioro, ``Max-Min Fairness and Its Applications to Routing and Load-Balancing in Communication
  Networks: A Tutorial'', \emph{IEEE Comm. Surveys and Tutorials}, vol. 10, no. 4, 2008.

\bibitem{neighborgood} Neighborgoods, http://neighborgoods.net/\,.

\bibitem{nytimes-sharing} NY Times, ``It's Not Just Nice to Share, It's the Future'', Jun. 2013.


\bibitem{opengarden} Open Garden, 2013, http://opengarden.com/\,.

\bibitem{osborne} M. J. Osborne, A. Rubinstein, ``A Cource in Game Theory'', \emph{MIT Press}, 1994.

\bibitem{saad-smart} W. Saad, Z. Han, and V. H. Poor, ``Coalitional Game Theory for Cooperative Micro-Grid Distribution Networks'', \emph{in Proc. of IEEE ICC}, 2011.


\bibitem{shapley-scarf} \rev{L. Shapley, and H. Scarf, ``On Cores and Indivisibility'', \emph{Journal of Math. Econ.}, vol. 1, 1974.}

\bibitem{swapit} \rev{Swap, www.swap.com}

\bibitem{Sofia-UPN} R. Sofia, et al., ``User-Provided Networks: Consumer as Provider'', \emph{IEEE Comm. Mag.}, 46(12), 2008.

\bibitem{unver-book} T. Sonmez, and M. U. Unver, ``Matching, Allocation, and Exchange of Discrete Resources'', \emph{in Handbook of Social Economics}, vol. 1A, 2011, pp. 781-852.

\bibitem{bewifi} Telefonica BeWiFi, 2014, http://www.bewifi.es/



\bibitem{boudecfairness07} B. Radunovic, and J. Y. Le Boudec, ``A Uunified Framework for Max-min and Min-max Fairness with Applications'', \emph{IEEE/ACM Trans. on Net.}, 15(5), 2007.

\bibitem{krishna06} Y. Wang, and A. Krishna, ``Timeshare Exchange Mechanisms'', \emph{Management Science}, vol. 52, no. 8, 2006, pp. 1223-1237.

\bibitem{zhang-proportional} F. Wu, and L. Zhang, ``Proportional Response Dynamics Leads to Market Equilibrium'', \emph{in Proc. of IEEE FOCS}, 2010.

\bibitem{wu-mesh} F. Wu, et al., ``Incentive-Compatible Opportunistic Routing for Wireless Networks'', \emph{in Proc. of ACM Mobicom}, 2008.


\end{thebibliography}

\end{document}